\documentclass[journal]{IEEEtran}
\usepackage{cite}
\usepackage{amsmath,amssymb,amsfonts}
\usepackage{graphicx}
\usepackage{algorithm,algpseudocode}
\usepackage[hidelinks]{hyperref}
\usepackage{textcomp}

\usepackage[dvipsnames]{xcolor}
 
\usepackage{mathtools}
\usepackage{amsthm}
\usepackage{abbreviations}
\usepackage{aligned-overset}
\usepackage{siunitx}
\usepackage{tikz,pgfplots}
\pgfplotsset{compat=1.18}
\usepackage{nicefrac}

\newtheorem{theorem}{Theorem}

\newtheorem{proposition}[theorem]{Proposition}

\newtheorem{remark}[theorem]{Remark}
\newtheorem{assumption}[theorem]{Assumption}

\newtheorem{procedure}[theorem]{Procedure}

\begin{document}



\title{
    Koopman-based feedback design\\
    with stability guarantees
}
\author{Robin Strässer,
    Manuel Schaller,
    Karl Worthmann,
    Julian Berberich,
    Frank Allgöwer
\thanks{F.\ Allgöwer is thankful that this work was funded by the Deutsche Forschungsgemeinschaft (DFG, German Research Foundation) under Germany's Excellence Strategy -- EXC 2075 -- 390740016 and within grant AL 316/15-1 -- 468094890.
K.\ Worthmann gratefully acknowledges funding by the German Research Foundation (DFG; project number 507037103).
R.\ Strässer thanks the Graduate Academy of the SC SimTech for its support.}%
\thanks{R.\ Strässer, J.\ Berberich, and F.\ Allgöwer are with the Institute for Systems Theory and Automatic Control, University of Stuttgart, 70550 Stuttgart, Germany (e-mail: \{robin.straesser, julian.berberich, frank.allgower\}@ist.uni-stuttgart.de)}%
\thanks{M.\ Schaller and K.\ Worthmann are with the Institute for Mathematics, Technische Universität Ilmenau, 99693 Ilmenau, Germany (e-mail: \{manuel.schaller, karl.worthmann\}@tu-ilmenau.de).}}

\maketitle
\begin{abstract}
    We present a method to design a state-feedback controller ensuring exponential stability for nonlinear systems using only measurement data. 
    Our approach relies on Koopman-operator theory and uses robust control to explicitly account for approximation errors due to finitely many data samples.
    To simplify practical usage across various applications, we provide a tutorial-style exposition of the feedback design and its stability guarantees for single-input systems.
    Moreover, we extend this controller design to multi-input systems and more flexible nonlinear state-feedback controllers using gain-scheduling techniques to increase the guaranteed region of attraction.
    As the proposed controller design is framed as a semidefinite program, it allows for an efficient solution.
    Further, we enhance the geometry of the region of attraction through a heuristic algorithm that establishes a connection between the employed Koopman lifting and the dynamics of the system.
    Finally, we validate the proposed feedback design procedure by means of numerical examples.
\end{abstract}

\section{Introduction}
\label{sec:introduction}

\IEEEPARstart{C}{ontrol} theory has recently encountered a shift towards data-driven methods by the increasing availability of measurement data. 
This trend is particularly prominent in the context of nonlinear systems, where classical approaches, such as system identification and first-principles modeling, prove time-consuming, require expert knowledge, and often lead to non-convex optimization problems~\cite{martin:schon:allgower:2023b}.
Data-driven methods may circumvent these challenges by exploiting available data to construct a suitable representation of the underlying complex system that allows for a convex controller design.

While data-driven feedback design for linear systems with guarantees for the closed-loop system has been extensively investigated in the literature~\cite{willems:rapisarda:markovsky:demoor:2005,depersis:tesi:2019, berberich:koch:scherer:allgower:2020,waarde:eising:trentelman:camlibel:2020}, obtaining rigorous guarantees for general nonlinear systems remains an open challenge, see also, e.g.,~\cite{hou:wang:2013,markovsky:dorfler:2021,faulwasser:ou:pan:schmitz:worthmann:2023,martin:schon:allgower:2023b} and the references therein.
Available stability guarantees for particular classes of nonlinear systems include, e.g., Hammerstein and Wiener systems~\cite{berberich:allgower:2020}, bilinear systems~\cite{bisoffi:depersis:tesi:2020a}, polynomial systems~\cite{guo:depersis:tesi:2022a,dai:sznaier:2021}, and rational dynamics~\cite{strasser:berberich:allgower:2021}, where the resulting controller designs are based on semidefinite programs (SDP).
For more general classes of nonlinear systems, an SDP-based controller design can be achieved through approaches such as polynomial approximation~\cite{martin:allgower:2023,martin:schon:allgower:2023a}, Kernel regression~\cite{fiedler:scherer:trimpe:2021,devonport:yin:arcak:2020}, or linear parameter-varying (LPV) embeddings~\cite{verhoek:koelewijn:haesaert:toth:2023,miller:sznaier:2023}.
A comprehensive overview of data-driven control methods for nonlinear systems is provided in~\cite{martin:schon:allgower:2023b}. 

B.O.\ Koopman introduced a different perspective in his seminal paper of 1931~\cite{koopman:1931}, offering an exact description of autonomous nonlinear dynamics through a linear, yet infinite-dimensional system.
This approach views the system through the lens of observables evolving linearly over time, bridging the gap between nonlinear and linear control methods. 
Over the past two decades, the Koopman operator has gained widespread popularity~\cite{mezic:2005,proctor:brunton:kutz:2018,otto:rowley:2021,mauroy:mezic:susuki:2020,bevanda:sosnowski:hirche:2021}.
For control-affine systems, the Koopman framework yields an infinite-dimensional bilinear operator, and a finite-dimensional matrix approximation is typically used for computational tractability. 
Common approximation techniques include extended dynamic mode decomposition (EDMD)~\cite{williams:kevrekidis:rowley:2015,korda:mezic:2018b,haseli:cortes:2023a} and machine learning methods~\cite{takeishi:kawahara:yairi:2017,otto:rowley:2019} (and the references therein).
A probabilistic analysis of the resulting approximation error is given in~\cite{chen:vaidya:2019,nuske:peitz:philipp:schaller:worthmann:2023,schaller:worthmann:philipp:peitz:nuske:2023} and was recently extended to kernel EDMD~\cite{philipp:schaller:worthmann:peitz:nuske:2024,kohne:philipp:schaller:schiela:worthmann:2024}.
Although the majority of the existing approaches predominantly assume linear lifted dynamics or neglect the deviation between the infinite-dimensional lifted system and its finite-dimensional approximation~\cite{brunton:brunton:proctor:kutz:2016,korda:mezic:2018a,sinha:nandanoori:drgona:vrabie:2022}, Koopman-based control has demonstrated success in a wide range of applications, cf.~\cite{bruder:fu:gillespie:remy:vasudevan:2021,kim:quan:chung:2023,budisic:mohr:mezic:2012} and \cite[Chapter III]{mauroy:mezic:susuki:2020}. However, closed-loop guarantees can only be deduced by incorporating the approximation error into the controller design and its analysis. 
The work in~\cite{strasser:berberich:allgower:2023a} introduces a robust linear controller design for an error-perturbed bilinear lifted system in discrete time based on a linear fractional representation (LFR). Under the assumption of a given finite-gain bound on the approximation error resulting from a finite dictionary and a finite data set, this approach provides closed-loop guarantees for the underlying nonlinear system. 
However, the practical application is challenging, since the work does not provide a procedure on how to ensure the required finite-gain bound.

In this work, we introduce a robust controller design that rigorously ensures end-to-end stability for nonlinear systems using measurement data only. 
Based on previous work on the approximation error~\cite{nuske:peitz:philipp:schaller:worthmann:2023,bold:grune:schaller:worthmann:2023}, we establish a novel error bound for the Koopman generator, which is proportional to the state and control in dependence on the amount of data. 
This allows us to extend the controller design from our previous works~\cite{strasser:berberich:allgower:2023a,strasser:berberich:allgower:2023b} to continuous time and solve the key technical challenge of stabilizing the original nonlinear system rather than some lifted bilinear one.
Notably, this is possible without having to switch between different coordinates but directly in terms of the \emph{true} system state.
Our results rely on an SDP in terms of linear matrix inequalities (LMI), which are efficiently solvable by standard software.
The presented controller design establishes a connection between a probabilistic tolerance concerning closed-loop guarantees for the nonlinear system and the necessary amount of data needed to satisfy a predetermined bound on the approximation error of the lifted system.
To enhance accessibility and implementation in a variety of applications, we provide a tutorial-style exposition of the proposed feedback design and its provable guarantees for a simplified setting with single-input systems.
Moreover, we generalize the controller design to more flexible nonlinear state-feedback controllers, reducing conservatism and resulting in guarantees within a larger region of attraction (RoA). 
We evaluate the presented controller by means of various numerical examples, e.g., by illustrating the RoA in the original state space.
The recent work~\cite{eyuboglu:powell:karimi:2024} is related to our approach since it proposes a robust approach in discrete time using the Koopman operator and an LPV formulation, but the derived stability guarantees hold only if the LPV system is already open-loop stable.

This work is structured as follows.
After providing the preliminaries for our approach in Section~\ref{sec:preliminaries}, we derive a data-driven controller design for nonlinear systems using convex optimization. 
Section~\ref{sec:controller-design-simple} provides the derivation of the proposed state-feedback controller design for single-input systems alongside a proof of closed-loop stability.
We consider multi-input systems and a more general controller parametrization in Section~\ref{sec:controller-design-general}.
The geometry of the guaranteed RoA induced by the nonlinear lifting and the controller design is discussed in Section~\ref{sec:geometry-of-RoA}.
Finally, we apply the developed controller design to numerical examples in Section~\ref{sec:numerical-examples}, before Section~\ref{sec:conclusion} concludes the paper.

\emph{Notation:~}
We write $I_p$ for the $p\times p$ identity matrix and $0_{p\times q}$ for the $p\times q$ zero matrix, where we omit the index if the dimension is clear from the context. 
If $A$ is symmetric, then we write $A\succ 0$ or $A\succeq 0$ if $A$ is positive definite or positive semidefinite, respectively. 
Negative (semi)definiteness is defined analogously. 
Matrix blocks which can be inferred from symmetry are denoted by $\star$ and we abbreviate $B^\top A B$ by writing $[\star]^\top AB$. 
We denote the Kronecker product by~$\kron$.
We write $[n:m]$ for the interval of integers $[n,m] \cap \mathbb{Z}$ and use the short-hand notation $[m] := [1:m]$.
A function $f:\bbR^p\to\bbR^q$ belongs to the differentiability class $\cC^k(\bbF,\bbH)$ if the function is $k$ times continuously differentiable on the domain $\bbF\subseteq\bbR^p$ mapping to $\bbH\subseteq\bbR^q$. 
Finally, we define the function class $\cK$ containing all functions $\alpha: [0,\infty) \to [0,\infty)$ which are continuous, strictly increasing, and satisfy $\alpha(0)=0$.
\section{Preliminaries}\label{sec:preliminaries}

In this section, we first introduce the 
problem setting in~Section~\ref{sec:problem-setting} and, then, derive a data-driven representation of nonlinear systems using the Koopman operator in~Section~\ref{sec:DD-system-representation}.

\subsection{Problem setting}
\label{sec:problem-setting}

In this paper, we consider \emph{unknown} nonlinear control-affine systems of the form 
\begin{equation}\label{eq:dynamics-nonlinear}
    \dot{x}(t) = f(x(t)) + \sum_{i=1}^mg_i(x(t)) u_i(t),
\end{equation}
where $x(t) \in\bbR^n$ denotes the state at time $t \geq 0$ and the control function $u \in L_\mathrm{loc}^\infty([0,\infty),\bbR^m)$ serves as an input. The map $f:\bbR^n\to\bbR^n$ is called drift, while $g_i:\bbR^n\to\bbR^n$, $i \in [n]$, are called input maps. 
For initial condition $x(0) = \hat{x}\in\bbR^n$ and control function $u$, we denote the solution of~\eqref{eq:dynamics-nonlinear} at time~$t$ by $x(t;\hat{x},u)$ tacitly assuming its existence and uniqueness.
We assume that $f(0)=0$ holds, i.e., that the origin is a controlled equilibrium for $u=0$.

The functions $f$ and $g_i$, $i \in [m]$, are unknown and only data samples $f(x_j) + \sum_{i=1}^m g_i(x_j)u_i$ for suitably chosen inputs~$u_i$ and sampling points $x_j \in \mathbb{X} \subset \mathbb{R}^n$, $j \in [d]$ with $d \in \bbN$, are available. Here, $\mathbb{X}$ is assumed to be compact and convex.
Further, we assume that the control values are restricted to a compact set $\mathbb{U}\subset \bbR^m$ with $0\in \operatorname{int}(\mathbb{U})$.
Our goal is the systematic design of a static state-feedback control law $\mu: \mathbb{R}^n \rightarrow \mathbb{R}^m$ such that the origin is exponentially stable w.r.t.\ the closed loop, i.e.,  
\begin{align}\label{eq:dynamics-nonlinear-feedback}
    \dot{x}(t) = f(x(t)) + \sum_{i=1}^m g_i(x(t)) \mu_i(x(t)).
\end{align}
In particular, we want to guarantee exponential stability in a later defined RoA in terms of a Lyapunov-function sublevel set.

The goal of this work is to construct a data-driven controller relying solely on the aforementioned samples of the system. To this end, we represent the nonlinear system~\eqref{eq:dynamics-nonlinear} using the Koopman-operator framework.
\subsection{Data-driven system representation using the Koopman operator}\label{sec:DD-system-representation}

In the following, we first introduce the Koopman-operator framework and its finite-dimensional approximation using data samples. Then, we discuss suitable error bounds on the resulting approximation error.

\subsubsection{Koopman operator and its approximation}

The Koopman operator $\cK_t^u$ corresponding to the system dynamics~\eqref{eq:dynamics-nonlinear} with constant control input $u(t) \equiv\ u \in \bbR^m$ is defined as 
\begin{equation}
    (\cK_t^u \varphi)(\hat{x}) = \varphi(x(t;\hat{x},u))
\end{equation}
for all $t\geq 0$, $\hat{x}\in\bbX$, $\varphi \in L^2(\bbX,\bbR)$, where the real-valued functions $\varphi$ are called \emph{observables}~\cite{koopman:1931}.
Here, we tacitly assumed invariance of $\bbX$ under the flow such that the observable functions are defined on $x(t;\hat{x},u)$ for all $\hat{x}\in \bbX$. This assumption can be relaxed by considering initial values contained in a tightened version of $\bbX$, see, e.g., \cite{goor:mahony:schaller:worthmann:2023} for details and is not needed in the following, but simplifies our brief recap of known results.
For this setting, $(\cK_t^u)_{t\geq 0}$ is a strongly-continuous semigroup of bounded linear operators. We define the respective infinitesimal generator~$\cL^u$ via 
\begin{equation}
    \cL^u \varphi \coloneqq \lim_{t\searrow 0} \frac{\cK_t^u\varphi - \varphi}{t}
    \quad 
    \forall\, \varphi \in D(\cL^u),
\end{equation}
where the domain $D(\cL^u)$ consists of all $L^2$-functions for which the above limit exists. By definition of this generator, the propagated observable satisfies $\dot{\varphi}(x)=\cL^u \varphi(x)$.
A key structural property of this generator is that it preserves control affinity~\cite{williams:hemati:dawson:kevrekidis:rowley:2016,surana:2016}, i.e., for $u\in\bbR^m$ we obtain
\begin{equation}\label{eq:Koopman-generator}
    \cL^u = \cL^0 + \sum_{i=1}^m u_i (\cL^{e_i}-\cL^0),
\end{equation}
where $\cL^0$ and $\cL^{e_i}$, $i\in[m]$, are the generators of the Koopman semigroups corresponding to the constant control functions $u(t) \equiv 0$ and $u(t) \equiv e_i$ with the unit vector $e_i$, $i\in[m]$, respectively. 

A data-driven approximation of the Koopman operator or its generator can be constructed via EDMD~\cite{williams:kevrekidis:rowley:2015} or generator EDMD~\cite{klus:nuske:peitz:niemann:clementi:schutte:2020}. Therein, one learns a finite-dimensional version of the operators restricted onto the span of a finite number of observable functions, so-called compressions, by means of a finite number of samples. 
In EDMD, the approximation error usually consists of two sources: A deterministic projection error due to finitely many observable functions and a probabilistic estimation error due to finitely many data points, cf.~\cite{korda:mezic:2018b} for the infinite-data limit and \cite{nuske:peitz:philipp:schaller:worthmann:2023,schaller:worthmann:philipp:peitz:nuske:2023} for the respective extensions to control systems and finite-data error bounds. 
An alternative route avoiding the a-priori choice of a dictionary is given by kernel EDMD, cf.~\cite{klus:nuske:hamzi:2020,williams:rowley:kevrekidis:2016}, see \cite{philipp:schaller:worthmann:peitz:nuske:2024} and, in particular,~\cite{kohne:philipp:schaller:schiela:worthmann:2024} for uniform finite-data error bounds.

We define the dictionary $\bbV\coloneqq \spn{\phi_k}_{k=0}^{N}$ representing the $(N+1)$-dimensional subspace spanned by the chosen observables $\phi_k:\bbR^n\to\bbR$, $k\in[0:N]$.
We include a constant function and the full state in the observables, i.e., we define $\Phi: \mathbb{R}^n \to\bbR^{N+1}$ by $\phi_{0}(x)\equiv 1$ and $\phi_k(x) = x_k$ for $k\in [n]$ resulting in
\begin{equation}\label{eq:lifting-function}
    \Phi(x) = \begin{bmatrix}
        1 & x^\top & \phi_{n+1}(x) & \cdots & \phi_{N}(x)
    \end{bmatrix}^\top,
\end{equation}
where $\phi_k\in\cC^2(\bbR^n,\bbR)$ satisfy $\phi_k(0)=0$, $k\in[n+1:N]$.
Note that the constant observable $\phi_{0}\equiv 1$ is included to account for state-independent input maps $g_i(x)=b_i$, $b_i\in\bbR^n$, $i\in [m]$.
\begin{remark}
    The consideration of dictionaries not containing the coordinate maps is an important issue if full state measurements are not available.
    In this case, when performing Koopman-based prediction and control, one has to either resort to reprojection techniques~\cite{goor:mahony:schaller:worthmann:2023} to obtain a lifted state which is consistent with the structure of the dictionary, i.e., lies on $\Phi(\mathbb{X}) \coloneqq \{ \Phi(x) \,|\, x \in \mathbb{X} \}$, or use delay embeddings~\cite{otto2024learning}.
\end{remark}

As the lifting~$\Phi$ defined in~\eqref{eq:lifting-function} explicitly contains the full state $x$ and satisfies $\phi_k(0) = 0$ for all $k \in [N]$, we get the lower bound
\begin{equation}\label{eq:philower}
    \|\Phi(x)-\Phi(0)\|^2 = \|x\|^2 + \sum_{k=n+1}^N\|\phi_k(x)\|^2 \geq \|x\|^2.
\end{equation}
Further, a corresponding upper bound can be deduced via local Lipschitz continuity of the lifting function $\Phi(x)$, which is a direct consequence of $\Phi \in \cC^1(\bbR^n,\bbR^N)$.
In particular, there exists an $L_{\Phi}>0$ such that
\begin{equation}\label{eq:phiupper}
    \|\Phi(x) - \Phi(0)\|\leq L_\Phi \|x\|
\end{equation}
holds for all $x\in\bbX$.

Throughout the paper, we assume that the dictionary $\bbV$ is invariant w.r.t.\ the dynamics in the following sense.
\begin{assumption}[Invariance of the dictionary $\bbV$]\label{ass:invariance-of-dictionary}
    For any $\phi\in\bbV$, the relation $\phi(x(t;\hat{x},u))\in\bbV$ holds for all $u(t)\equiv\hat{u}\in\bbU$, $\hat{x}\in\bbX$, and $t\geq 0$.
\end{assumption}

We briefly discuss this assumption, cf.~also the paragraph after the proof of Proposition~\ref{lm:dynamics-lifted}.
We denote the $L^2$-orthogonal projection onto~$\bbV$ 
by $P_\bbV$. Then, Assumption~\ref{ass:invariance-of-dictionary} can be formulated compactly as 
\begin{equation}\label{eq:generator_matrix_coincides}
    P_\bbV\left.\cL^u\right|_\bbV = \left.\cL^u\right|_\bbV,
\end{equation}
i.e., the projection of $\left.\cL^u\right|_\bbV$ (the Koopman generator restricted to $\bbV$) onto $\bbV$ is equal to $\left.\cL^u\right|_\bbV$.
As a consequence, we obtain the Koopman operator for $u\in\bbU$ via $\cK_{t}^u|_\bbV = e^{t \cL^u|_\bbV}$.
Assumption~\ref{ass:invariance-of-dictionary} is commonly assumed when representing controlled systems via the Koopman operator.
Sufficient conditions for (approximately) satisfying this assumption are provided, e.g., by~\cite{goswami:paley:2021,brunton:brunton:proctor:kutz:2016,korda:mezic:2020}.
If the dictionary $\bbV$ does not fulfill the assumption, it is also possible to derive error bounds for the projection error~\cite{schaller:worthmann:philipp:peitz:nuske:2023}.
We note that~\cite{iacob:toth:schoukens:2024} proposes to consider also a nonlinear dependence of the lifted dynamics on the input signal, whereas we only require state-dependent observables for our proposed controller. Further,~\cite{haseli:cortes:2023b} introduces a notion of proximal invariance using Jordan principal angles, which provides a promising route for formulating approximate invariance in future work.
Follow-up work should be devoted to exploring structured approaches of inferring the observables for dynamical control systems, e.g., via kernels~\cite{philipp:schaller:worthmann:peitz:nuske:2023b}.

We assume that the nonlinear dynamics~\eqref{eq:dynamics-nonlinear} are unknown and we have state-derivative data $\{ x_j^{\bar{u}},\dot{x}_j^{\bar{u}}\}_{j=1}^{d^{\bar{u}}}$ for the constant control inputs $u(t) \equiv \bar{u}$, where $\bar{u}\in\{0,e_1,...,e_m\}$ and $d^{\bar{u}}$ denotes the number of data points for the control vector~$\bar{u}$.
Since the lifting functions are known, this allows constructing data samples $\lbrace \phi_k(x_j^{\bar{u}}), \langle \nabla \phi_k(x_j^{\bar{u}}), \dot{x}_j^{\bar{u}} \rangle \rbrace$, $k \in [0:N]$ and $j\in [d^{\bar{u}}]$.
Throughout this paper, we assume that derivative measurements are available, which is a standard requirement in continuous-time data-driven control.
It is also possible to derive analogous results for both continuous-time and discrete-time systems without resorting to derivative measurements, see~\cite{strasser:schaller:worthmann:berberich:allgower:2024} for the required modifications. 
This is particularly interesting for future research on noisy data.

We emphasize that we require sufficiently many snapshots for each of the $m+1$ open-loop dynamics corresponding to a constant control~$\bar{u}\in\{0,e_1,...,e_m\}$.
For the data generation, the control inputs $e_i$, $i\in[m]$, may be chosen as any basis of $\bbR^m$, e.g., since $0\in \operatorname{int}(\bbU)$, as scaled unit vectors to ensure their feasibility.

We briefly recall the consequences of the structure of the Koopman generator and the chosen dictionary and show how to design a suitable data-driven surrogate model as proposed in~\cite[Sec.~3.1]{bold:grune:schaller:worthmann:2023}.
The Koopman generator $\cL^u$ corresponding to the chosen dictionary defined in~\eqref{eq:lifting-function} has a particular structure due to the presence of the constant observable $\phi_0(x)\equiv 1$ satisfying
\begin{equation}\nonumber
    \tfrac{\mathrm{d}}{\mathrm{d}t}\phi_0(x(t;\cdot,u)) = \langle \phi_0(x(t;\cdot,u)), f(x(t;\cdot,u),u(t)) \rangle \equiv 0.    
\end{equation}
As $\ddt{}\phi_0(x(t;\cdot,u))$ corresponds to the first row of $\cL^u$, we may infer that the first row of $\cL^u$ can only have zero entries for all $u\in\bbU$.
Hence, one may deduce that $\cL^0$ and $\cL^{e_i}$, $i\in[m]$, contain only zeros in their first row.
Moreover, due to the assumption that the drift term of the dynamics~\eqref{eq:dynamics-nonlinear} vanishes at the origin, i.e., $f(0)=0$, we compute
\begin{equation}
    (\cL^0 \Phi)(0) =  \nabla \Phi(0)^\top f(0) = 0,
\end{equation}
i.e., in $\cL^0$ there are no contributions from the constant observable $\phi_0(x)\equiv 1$.
Hence, the first column of $\cL^0$ must be zero.
Thus, the generators are of the form
\begin{equation}\label{eq:Koopman-generator-structure}
        \cL^0=\begin{bmatrix} 
            0 & 0_{1\times N} \\
            0_{N\times 1} & \cL^0_{22}
        \end{bmatrix},
        \quad
        \cL^{e_i}=\begin{bmatrix}
            0 & 0_{1\times N} \\
            \cL^{e_i}_{21} & \cL^{e_i}_{22}
        \end{bmatrix},
\end{equation}
with $\cL^0_{22} \in \bbR^{N\times N}$ and $\cL^{e_i}_{21}\in \bbR^{N\times 1}$,   $\cL^{e_i}_{22}\in \bbR^{N\times N}$, $i\in [m]$, respectively. 

In the following, we show how to enforce this structure of the Koopman generator also for its data-driven surrogate model.
Here, we use EDMD to obtain a data-based approximation $\cL^u_d$ of the true Koopman generator $\cL^{u}$.
Based on $d^{\bar{u}}$ data points for $\bar{u}\in\{0,e_1,...,e_m\}$ and the dictionary $\bbV$, we define the ($(N+1)\times d^{\bar{u}}$)-matrices $X^{e_i}$ for $i\in[m]$ and the ($N\times d^{\bar{u}}$)-matrices $X^0$, $Y^{\bar{u}}$ for $\bar{u}\in\{0,e_1,...,e_m\}$, where
\begin{subequations}\label{eq:data-matrices}
    \begin{align}
    X^{e_i} &\coloneqq \begin{bmatrix}
        \Phi(x_1^{e_i})
        & \cdots & 
        \Phi(x_{d^{e_i}}^{e_i})
    \end{bmatrix},
    \\
    X^{0} &\coloneqq 
    \begin{bmatrix}
        0_{N\times 1} & I_N
    \end{bmatrix}
    \begin{bmatrix}
        \Phi(x_1^{0}) & \cdots & \Phi(x_{d^0}^{0})
    \end{bmatrix},
    \\
    Y^{\bar{u}} &\coloneqq 
    \begin{bmatrix}
        0_{N\times 1} & I_N
    \end{bmatrix}
    \begin{bmatrix}
        \cL^{\bar{u}}\Phi(x_1^{\bar{u}}) & \cdots & \cL^{\bar{u}}\Phi(x_{d^{\bar{u}}}^{\bar{u}})
    \end{bmatrix},
\end{align}
\end{subequations}
with $\cL^{\bar{u}}\Phi(x_j^{\bar{u}}) = \begin{bmatrix} 
    (\cL^{\bar{u}} \phi_0)(x_j^{\bar{u}}) & \cdots & (\cL^{\bar{u}} \phi_N)(x_j^{\bar{u}})
\end{bmatrix}^\top$ and $
    (\cL^{\bar{u}} \phi_k)(x_j^{\bar{u}})
    = \grad \phi_k(x_j^{\bar{u}})^\top \dot{x}_j^{\bar{u}}
$
for $k\in[0:N]$ and $j\in[d^{\bar{u}}]$. 
Then, the generator EDMD-based surrogate for the Koopman generator $\left.P_\bbV \cL^{\bar{u}}\right|_\bbV$ is given by
\begin{equation}\label{eq:Kooman-generator-approximation}
    \cL^0_d=\begin{bmatrix} 
        0 & 0_{1\times N} \\
        0_{N\times 1} & A
    \end{bmatrix},
    \
    \cL^{e_i}_d=\begin{bmatrix}
        0 & 0_{1\times N} \\
        B_{0,i} & \hat{B}_i
    \end{bmatrix}
\end{equation}
for $i\in[m]$, where 
\begin{subequations}\label{eq:EDMD-optimization}
    \begin{align}
        A 
        &= \argmin_{A\in\bbR^{N\times N}} 
        \left\| Y^0 - A X^0 \right\|_\mathrm{F},
        \\
        \begin{bmatrix}B_{0,i} & \hB_i\end{bmatrix} 
        &= \argmin_{\substack{
            B_{0,i}\in\bbR^{N}, \\
            \hB_i\in\bbR^{N\times N}
        }}
        \left\|
            Y^{e_i} 
            - \begin{bmatrix}B_{0,i} & \hB_i\end{bmatrix} X^{e_i}
        \right\|_\mathrm{F}
    \end{align}
\end{subequations}
for $i\in[m]$ and $\|\cdot\|_\mathrm{F}$ denotes the Frobenius norm. 
The explicit solutions to these problems read $
    A = Y^{0} (X^0)^\dagger
$ and $
    \begin{bmatrix}B_{0,i} & \hB_i\end{bmatrix} = Y^{e_i}(X^{e_i})^\dagger
$.
Motivated by the control-affine structure of the Koopman generator~\eqref{eq:Koopman-generator}, we define its data-based approximation as 
\begin{equation}
    \cL_d^u = \cL_d^0 + \sum_{i=1}^m u_i(\cL_d^{e_i} - \cL_d^0).
\end{equation}

\subsubsection{EDMD error bounds}
The following proposition captures the error between the Koopman generator and its data-driven EDMD estimate.
\begin{proposition}[\!\!{\cite[Thm.~3]{schaller:worthmann:philipp:peitz:nuske:2023}}]\label{prop:sampling-error-bound-Koopman-generator}
    Suppose that Assumption~\ref{ass:invariance-of-dictionary} holds and the data samples are i.i.d.
    Further, let an error bound $c_r>0$ and a probabilistic tolerance $\delta \in (0,1)$ be given. Then, there is an amount of data $d_0\in\bbN$ such for all $d\geq d_0$ and all $u\in\bbU$ we have the error bound 
    \begin{equation}\label{eq:sampling-error-bound-Koopman-generator}
        \|\cL^u\vert_\bbV - \cL_d^u\| \leq c_r
    \end{equation}
    with probability $1-\delta$.
\end{proposition}

\begin{remark}
    The sufficient amount of data~$d_0$ was derived in~\cite[Thm.~3]{schaller:worthmann:philipp:peitz:nuske:2023}. 
    For given dictionary size $N+1$, probabilistic tolerance $\delta \in (0,1)$, and error bound $c_r>0$, we define $\tilde{\delta} := \frac {\delta}{3(m+1)} \in (0,1)$ and
    \begin{equation*}
        \tilde{c}_r 
        \coloneqq 
        \frac{c_r}{(m+1)\left(1  + \max_{u \in \mathbb{U}} \sum_{i=1}^{m}|u_i| \right)} > 0.
    \end{equation*}
    Let matrices $A^{(k)},C \in \mathbb{R}^{(N+1) \times (N+1)}$, $k\in[0:m]$, be defined by $\left(A^{(k)}\right)_{i,j} = \langle \phi_i,\mathcal{L}^{e_k} \phi_j \rangle_{L^2(\mathbb{X})}$ and $C_{i,j} := \langle \phi_i,\phi_j\rangle_{L^2(\mathbb{X})}$ and set
    \begin{equation*}
        \tilde{c}_{r,k} 
        = \min\left\{1,\frac{1}{\|A^{(k)}\|\|C^{-1}\|}\right\} \cdot \frac {\|A^{(k)}\| \tilde{c}_r}{2\|A^{(k)}\|\|C^{-1}\| + \tilde{c}_r}.
    \end{equation*}
    Then, $d_0\in \mathbb{N}$ for the bound 
    of Proposition~\ref{prop:sampling-error-bound-Koopman-generator} is given by
    \begin{equation}\label{eq:mindata}
        d_0 \geq \max_{k=0,\ldots,m}\frac{(N+1)^2}{\tilde{\delta}\tilde{c}_{r,k}^2} \max\left\{\|\Sigma_{A^{(k)}} \|^2_F, \|\Sigma_{C}\|^2_F\right\}
    \end{equation}
    where $\Sigma_{A^{(k)}}$ and $\Sigma_{C}$ are variance matrices defined via
    \begin{align*}
        &\!\left(\Sigma_{A^{(k)}}\right)_{i,j}^2 
        =  \frac{1}{|\bbX|}\int_\mathbb{X} \phi_i(x)^2 \langle \nabla \phi_j(x), f(x) + g_k(x)\rangle^2\,\mathrm{d}x
        \\
        &\qquad\qquad\quad
        - \left(\frac{1}{|\bbX|}\int_\mathbb{X} \phi_i(x) \langle \nabla \phi_j(x), f(x) + g_k(x) \rangle \,\mathrm{d}x\right)^2\!,
        \\
        &\!\left(\Sigma_{C}\right)_{i,j}^2
        = \frac{1}{|\bbX|}\int_\mathbb{X} \phi_i(x)^2 \phi_j(x)^2\,\mathrm{d}x - \!\bigg[\frac{1}{|\bbX|}\int_\mathbb{X} \phi_i(x)\phi_j(x)\,\mathrm{d}x\bigg]^2
    \end{align*}
    for $(i,j)\in [N+1]^2$, where $|\bbX|$ denotes the Lebesgue measure of $\bbX$.
\end{remark}
In the following, we use the abbreviation $x(t) = x(t;\hat{x},u)$. 
Using the Koopman generator, the observables along the nonlinear dynamics~\eqref{eq:dynamics-nonlinear} satisfy
\begin{align}
    \ddt{}\Phi(x(t))
    &= \cL^u\Phi(x(t))
    \nonumber\\ \label{eq:dynamics-lifted-general}
    &= \cL_d^u \Phi(x(t)) + (\cL^u - \cL^u_d)\Phi(x(t)).
\end{align}
Thus, the dynamics along the data-driven surrogate $\cL^u_d$ can be interpreted as a perturbed version of the original lifted dynamics with the remainder
\begin{align}
    r(x,u) 
    &= (\cL^u - \cL^u_d)\Phi(x) 
    \nonumber\\\label{eq:remainder-plus-minus-Phi0}
    &= (\cL^u - \cL^u_d) (\Phi(x) - \Phi(0) + \Phi(0)),
\end{align}
which in view of Proposition~\ref{prop:sampling-error-bound-Koopman-generator} satisfies the bound 
\begin{equation}\label{eq:finite-gain-bound-remainder-offset}
    \|r(x,u)\| \leq c_r (\|\Phi(x) - \Phi(0)\| + \|\Phi(0)\|)
\end{equation} 
for all $x\in \bbX$ and $u\in \bbU$ with probability $1-\delta$.

This error bound, however, has a significant drawback in view of controller design:
The upper bound contains a state-independent part and thus, as $\Phi(0)\neq 0$ due to the constant observable, it does not vanish for $x=0$ and $u=0$, which is the targeted equilibrium point of the dynamics.
To address this issue by means of an alternative (and \emph{proportional}) error bound, we adapt the line of reasoning of~\cite[Prop.~8]{bold:grune:schaller:worthmann:2023} leading to new error estimates for the compression of the Koopman generator in the lifted space, which guarantees that the residual vanishes at the origin, i.e., $r(x,u)=0$ for $x=0$ and $u=0$.
To this end, we define the reduced lifted state 
$
    \hat{\Phi}(x)=\begin{bmatrix}0_{N\times 1} & I_N\end{bmatrix}\Phi(x)  
$
satisfying $\hat{\Phi}(0)=0$ which allows the characterization of the lifted dynamics with a proportional error bound on the remainder as follows.
\begin{proposition}\label{lm:dynamics-lifted}
    Suppose that Assumption~\ref{ass:invariance-of-dictionary} holds and the data samples are i.i.d. 
    Further, let a probabilistic tolerance $\delta\in(0,1)$ and an amount of data $d_0\in\bbN$ be given. Then, there is a constant $c_r\in\mathcal{O}(\nicefrac{1}{\sqrt{\delta d_0}})$ such that the lifted dynamics~\eqref{eq:dynamics-lifted-general} are equivalently captured by
    \begin{multline}\label{eq:dynamics-lifted}
        \ddt{}\hat{\Phi}(x(t)) 
        = A \hat{\Phi}(x(t)) 
        + B_0 u(t) 
        \\
        + \sum_{i=1}^m u_i(t)B_i\hat{\Phi}(x(t)) 
        + \hat{r}(x(t),u(t))
    \end{multline}
    with $
        B_0=\begin{bmatrix}
            B_{0,1}&\cdots & B_{0,m}
        \end{bmatrix}
    $, $B_i=\hB_i - A$, where the remainder is bounded by
    \begin{equation}\label{eq:finite-gain-bound-remainder}
        \|\hat{r}(x,u)\| \leq c_r(\|\hat{\Phi}(x)\| + \|u\|)
    \end{equation}
    for all $(x,u)\in \bbX \times \bbU$ with probability $1-\delta$.
\end{proposition}
\begin{proof}
    Let $(x,u)\in \bbX \times \bbU$. The structure of the Koopman generators~\eqref{eq:Koopman-generator-structure} ensures $
        (\cL^0\Phi)(x)=\left[\begin{smallmatrix}0\\(\cL^0_{22}\hat{\Phi})(x)\end{smallmatrix}\right]
    $ and $
        (\cL^{e_i}\Phi)(x)=\left[\begin{smallmatrix}0\\\cL^{e_i}_{21} + (\cL^{e_i}_{22}\hat{\Phi})(x)\end{smallmatrix}\right]
    $,
    $i\in[m]$.
    By defining $\cL^e_{21} = \begin{bmatrix} \cL^{e_1}_{21} & \cdots & \cL^{e_m}_{21}\end{bmatrix}$ and $\tilde{\cL}^{e_i}_{22}=\cL^{e_i}_{22}-\cL^0_{22}$ we deduce 
    \begin{equation*}
        (\cL^u\Phi)(x) = \begin{bmatrix}0\\(\cL^0_{22}\hat{\Phi})(x) + \cL^e_{21} u + \sum_{i=1}^m u_i (\tilde{\cL}^{e_i}_{22}\hat{\Phi})(x)\end{bmatrix}.
    \end{equation*}
    Due to~\eqref{eq:Kooman-generator-approximation} the same structure applies for the corresponding EDMD approximation, i.e.,
    \begin{equation}
        \cL^u_d\Phi(x) = \begin{bmatrix}0\\A\hat{\Phi}(x) + B_0 u + \sum_{i=1}^m u_i B_i\hat{\Phi}(x)\end{bmatrix}
    \end{equation}
    with $
        B_0=\begin{bmatrix}
            B_{0,1}&\cdots & B_{0,m}
        \end{bmatrix}
    $ and $B_i=\hB_i - A$.
    Since $\hat{\Phi}(0)=0$ holds, $(\cL^u - \cL^u_d)\Phi(0)=\left[\begin{smallmatrix}0\\(\cL^e_{21} - B_0)u\end{smallmatrix}\right]$ such that
    \begin{equation}\label{eq:proof:dynamics-lifted-reduced-input}
        \|(\cL^u - \cL^u_d)\Phi(0)\| \leq \|\cL^e_{21} - B_0\|\|u\|.
    \end{equation}
    Moreover, by construction, the first row of the remainder $r(x,u)$ in~\eqref{eq:dynamics-lifted-general} vanishes, i.e., 
    \begin{equation}
        r(x,u)=(\cL^u-\cL^u_d)\Phi(x) = \begin{bmatrix}0\\\hat{r}(x,u)\end{bmatrix},
    \end{equation}
    which implies that the first row of the right-hand side of~\eqref{eq:dynamics-lifted-general} vanishes. Together with the vanishing derivative of the constant function along the flow, i.e., $\frac{\mathrm{d}}{\mathrm{d}t} \phi_0(x(t)) \equiv 0$, this yields~\eqref{eq:dynamics-lifted} when omitting the trivial first line.
    Finally, we observe $\Phi(x)-\Phi(0) = \begin{bmatrix}0\\\hat{\Phi}(x)\end{bmatrix}$ and conclude 
    \begin{align}
        &\|\hat{r}(x,u)\| 
        = \|r(x,u)\|
        \nonumber\\
        \overset{\eqref{eq:remainder-plus-minus-Phi0}}&{\leq} \|\cL^u-\cL^u_d\|\|\Phi(x)-\Phi(0)\| + \|(\cL^u - \cL^u_d)\Phi(0)\|
        \nonumber\\
        \overset{\eqref{eq:proof:dynamics-lifted-reduced-input}}&{\leq} \|\cL^u-\cL^u_d\|\|\hat{\Phi}(x)\| + \|\cL^e_{21} - B_0\|\|u\|.
    \end{align}
    Further, as $e_i\in \mathbb{U}$ for all $i\in [m]$, then the estimate in~\eqref{eq:sampling-error-bound-Koopman-generator} not only bounds $\|\cL^u-\cL^u_d\|$ but yields also the (conservative) upper bound 
    \begin{equation*}
        \|\cL^e_{21} - B_0\| = \|(\cL^u - \cL^u_d)\Phi(0)\| \leq \|\cL^u - \cL^u_d\|\|\Phi(0)\|\leq c_r.
    \end{equation*}
    Hence, invoking the bound on $\|\cL^u - \cL^u_d\|$ of Proposition~\ref{prop:sampling-error-bound-Koopman-generator} yields the claim.
\end{proof}
Proposition~\ref{lm:dynamics-lifted} establishes a finite-dimensional \emph{bilinear} representation of the nonlinear system~\eqref{eq:dynamics-nonlinear}, where the remainder is bounded via the norm of the lifted state $\hat{\Phi}(x)$ and the input $u$.
Due to the construction of $\hat{\Phi}$, this removes the constant offset at the origin $x=0$ occurring in the error bound~\eqref{eq:finite-gain-bound-remainder-offset} and allows convex data-driven controller design for the nonlinear system~\eqref{eq:dynamics-nonlinear}.
Moreover, the constant observable $\phi_0\equiv 1$ results in the state-independent input mapping via $B_0$.
We note that, as becomes clear from the last part of the proof of Proposition~\ref{lm:dynamics-lifted}, the constant appearing with the control value $u$ in~\eqref{eq:finite-gain-bound-remainder} can be tightened further by exploiting the derivation of~\eqref{eq:sampling-error-bound-Koopman-generator}, which is left for future research.
Last, we briefly discuss Assumption~\ref{ass:invariance-of-dictionary} in view of our central result above. Due to the invariance of the dictionary, the Koopman generator coincides with a matrix, see~\eqref{eq:generator_matrix_coincides}. 
However, even if this does not hold, i.e., if $\mathcal{L}^u$ maps to a function space, we note that most of the arguments presented in the proof of Proposition~\ref{lm:dynamics-lifted} may be carried over upon utilizing, in addition, novel projection-error bounds~\cite{kohne:philipp:schaller:schiela:worthmann:2024}. 
In particular, the structure of the Koopman generator formulated in~\eqref{eq:Koopman-generator-structure} remains valid upon considering the partition of the function space into the span of the constant function $\phi_0$ and the reduced lifted state $\hat{\Phi}$.

For the remainder of the paper, we omit the time argument for ease of notation, wherever it is clear from the context.
\section{Data-driven control of nonlinear systems}\label{sec:controller-design-simple}
In this section, we develop a controller for the nonlinear system~\eqref{eq:dynamics-nonlinear} based on the data-driven Koopman representation derived in Section~\ref{sec:DD-system-representation}. 
We provide a tutorial-style exposition of the derivation in a simplified setting with single-input systems and state-feedback controllers that depend linearly on the lifted state $\hat{\Phi}(x)$, i.e., $\mu(x) = K \hat{\Phi}(x)$ holds with a matrix $K \in \mathbb{R}^{1 \times N}$.
Section~\ref{sec:reformulation-as-robust-linear-control-problem} introduces a reformulation of the corresponding lifted system dynamics as a robust linear control problem. 
Then, we use convex optimization to solve the control problem in Section~\ref{sec:solution-via-CO}.
The case with multi-input systems and more general controller parametrizations will be considered in Section~\ref{sec:controller-design-general}.

\subsection{Reformulation as robust linear control problem}\label{sec:reformulation-as-robust-linear-control-problem}
First, we recall the lifted system dynamics~\eqref{eq:dynamics-lifted} for scalar inputs ($m=1$), i.e., 
\begin{equation}\label{eq:dynamics-lifted-single-input}
    \ddt{}
    \hat{\Phi}(x) 
    = A \hat{\Phi}(x) + B_0 u + B_1 \hat{\Phi}(x) u + \hat{r}(x,u),
\end{equation}
where the nonlinear remainder $\hat{r}(x,u)$ is bounded by the proportional bound~\eqref{eq:finite-gain-bound-remainder} on $\bbX \times \bbU$.
Our goal is to design a state-feedback controller $u=\mu(x)$ such that the origin of the closed-loop system is exponentially stable.
To this end, we view~\eqref{eq:dynamics-lifted-single-input} as an uncertain bilinear system, where the uncertainty is given by $\hat{r}(x,u)$.
To cope with the uncertainty, we employ linear robust control techniques enabling methods from convex optimization.

First, we treat the remainder in the lifted dynamics~\eqref{eq:dynamics-lifted} as uncertainty and stabilize the system robustly using that $\hat{r}(x,u)$ satisfies the error bound~\eqref{eq:finite-gain-bound-remainder} for all $x\in\bbX,u\in\bbU$.
To this end, we define the remainder $\varepsilon:\bbR^N\times\bbR^m\to\bbR^N$ depending on the lifted state as 
\begin{equation}\label{eq:epsilon-definition}
    \varepsilon(v_1,v_2)=\hat{r}\left(\begin{bmatrix}I_n&0_{n\times N-n}\end{bmatrix}v_1,v_2\right).
\end{equation}
According to~\eqref{eq:finite-gain-bound-remainder}, $\varepsilon$ satisfies also a proportional bound, i.e.,
\begin{equation}\label{eq:finite-gain-bound-epsilon}
    \| \varepsilon(\hat{\Phi}(x),u)\| 
    = \|\hat{r}(x,u)\|
    \leq c_r(\|\hat{\Phi}(x)\| + \|u\|)
\end{equation}
for all $x\in\bbX,u\in\bbU$.
Second, we introduce an additional uncertainty $\Delta_\Phi$ corresponding to the term $\hat{\Phi}(x)$.
To this end, we need to ensure that $x$ and, in particular, $\hat{\Phi}(x)$ are bounded.
Thus, we introduce an a priori fixed set $\mathbf{\Delta}_\Phi$ of the form $
    \mathbf{\Delta}_\Phi \coloneqq \{\Delta_\Phi\in\bbR^N \mid \eqref{eq:condition-x-in-Z}~\text{holds}\}
$ 
with
\begin{equation}\label{eq:condition-x-in-Z}
    \begin{bmatrix}
        \Delta_\Phi \\ 1
    \end{bmatrix}^\top 
    \begin{bmatrix}
        Q_z & S_z \\ S_z^\top & R_z
    \end{bmatrix}
    \begin{bmatrix}
        \Delta_\Phi \\ 1
    \end{bmatrix}
    \geq 0
\end{equation}
for fixed $Q_z\in\bbR^{N\times N}$, $S_z\in\bbR^{N}$, $R_z\in\bbR$ with $Q_z \prec 0$ and $R_z > 0$ for which the inverse
$
    \left[\begin{smallmatrix}
        \tQ_z & \tS_z \\
        \tS_z^\top & \tR_z
    \end{smallmatrix}\right]
    \coloneqq 
    \left[\begin{smallmatrix}
        Q_z & S_z \\
        S_z^\top & R_z
    \end{smallmatrix}\right]^{-1}    
$
exists.
The parametrization~\eqref{eq:condition-x-in-Z} of $\mathbf{\Delta}_\Phi$ includes, e.g., a region described by $\Delta_\Phi^\top \Delta_\Phi\leq c$ with $c>0$ when choosing $Q_z=-I$, $S_z=0$, and $R_z=c$.
Our later derived controller design will ensure that the closed-loop system satisfies $\hat{\Phi}(x(t))\in\mathbf{\Delta}_\Phi$ for all $t$.
More precisely, we stabilize the bilinear system~\eqref{eq:dynamics-lifted-single-input} by robustly stabilizing
\begin{equation}\label{eq:dynamics-lifted-single-input-Delta}
    \ddt{}\hat{\Phi}(x) 
    = A \hat{\Phi}(x) + B_0 u + B_1 \Delta_\Phi u + \varepsilon(\hat{\Phi}(x),u)
\end{equation}
for all $\Delta_\Phi\in\mathbf{\Delta}_\Phi$ and perturbation functions $\varepsilon$ satisfying~\eqref{eq:finite-gain-bound-epsilon} for all $x\in\bbX,u\in\bbU$.

Now, we write the lifted system described by the uncertain bilinear system~\eqref{eq:dynamics-lifted-single-input} using a linear fractional representation (LFR). 
In particular, we consider the LFR corresponding to~\eqref{eq:dynamics-lifted-single-input-Delta} given by
\begin{subequations}\label{eq:LFR-nonlinear-open-loop-single-input}
    \begin{align}
        \begin{bmatrix}
            \ddt{}\hat{\Phi}(x) \\ u \\ \begin{bmatrix}v_1\\v_2\end{bmatrix}
        \end{bmatrix}
        &= \begin{bmatrix}
            A & B_0 & B_1 & I \\
            0 & I & 0 & 0 \\
            \begin{bmatrix} I \\ 0\end{bmatrix} & \begin{bmatrix}0\\I\end{bmatrix} & 0 & 0
        \end{bmatrix}
        \begin{bmatrix}
            \hat{\Phi}(x) \\ u \\ w_\Phi \\ w_r
        \end{bmatrix},\\
        w_\Phi &=\Delta_\Phi u, \label{eq:LFR-nonlinear-open-loop-single-input:uncertainty-bilinearity} \\
        w_r &= \varepsilon(v_1,v_2) \label{eq:LFR-nonlinear-open-loop-single-input:uncertainty-remainder}
    \end{align}
\end{subequations}
with $\Delta_\Phi\in\mathbf{\Delta}_\Phi$ and $\varepsilon$ satisfying~\eqref{eq:finite-gain-bound-epsilon} for all $x\in\bbX,u\in\bbU$.
An LFR as in~\eqref{eq:LFR-nonlinear-open-loop-single-input} is a common representation of uncertain systems~\cite{zhou:doyle:glover:1996}. 
Here, the LFR is exposed to two uncertainties, where we reduce the bilinear control problem to a linear control problem subject to an unknown state-dependent nonlinearity, i.e., the (known) lifted state $\Delta_\Phi=\hat{\Phi}(x)$, and the nonlinearity $\varepsilon(\hat{\Phi}(x),u)$ corresponding to the (unknown) remainder $\hat{r}(x,u)$.
Note that the LFR~\eqref{eq:LFR-nonlinear-open-loop-single-input} is highly structured with several zero terms and, thus, the dynamics depend linearly on the uncertainties. 
However, we note that the LFR framework offers the possibility to adapt the subsequent analysis to more general (fractional) uncertainty descriptions in future work.

\subsection{Solution via linear matrix inequalities}\label{sec:solution-via-CO}

In the following, we design a state-feedback controller that is linear in the lifted state, i.e., $\mu(x)=K\hat{\Phi}(x)$ for some $K\in\bbR^{1\times N}$.
Setting $A_K = A+B_0K$ and substituting the input by the feedback in the open-loop LFR~\eqref{eq:LFR-nonlinear-open-loop-single-input}, we obtain the corresponding closed-loop LFR
\begin{subequations}\label{eq:LFR-nonlinear-closed-loop-single-input}
    \begin{align}
        \begin{bmatrix}
            \ddt{}\hat{\Phi}(x) \\ \mu(x) \\ \begin{bmatrix}v_1\\v_2\end{bmatrix}
        \end{bmatrix}
        &= \begin{bmatrix}
            A_K & B_1 & I \\
            K & 0 & 0 \\
            \begin{bmatrix} I \\ K \end{bmatrix} & 0 & 0
        \end{bmatrix}
        \begin{bmatrix}
            \hat{\Phi}(x) \\ w_\Phi \\ w_r
        \end{bmatrix},\\
        w_\Phi &= \Delta_\Phi \mu(x), \\
        w_r &= \varepsilon(v_1,v_2)
    \end{align}
\end{subequations}
with $\Delta_\Phi\in\mathbf{\Delta}_\Phi$ and $\varepsilon$ satisfying~\eqref{eq:finite-gain-bound-epsilon} for all $x\in\bbX,u\in\bbU$.

The following theorem establishes a controller design method guaranteeing exponential stability for the nonlinear system~\eqref{eq:dynamics-nonlinear} by solving an LMI feasibility problem.
\begin{theorem}\label{thm:stability-condition-LFR-single-input}
    Let Assumption~\ref{ass:invariance-of-dictionary} hold. 
    Suppose a desired error bound $c_r>0$ and a probabilistic tolerance $\delta \in (0,1)$ in the sense of Proposition~\ref{lm:dynamics-lifted} are given.
    If there exist a matrix $0\prec P=P^\top\in \bbR^{N\times N}$, a matrix $L\in\bbR^{1\times N}$, and scalars $\lambda>0$, $\nu>0$, $\tau>0$
    such that~\eqref{eq:stability-condition-LFR-single-input}%
    \begin{figure*}[!t]
        \begin{equation}\label{eq:stability-condition-LFR-single-input}
            \begin{bmatrix}
                -A P - B_0L - PA^\top - L^\top B_0^\top
                -\tau I_{N}
                & -L^\top - \lambda B_1 \tS_z
                & -\begin{bmatrix}P & L^\top\end{bmatrix}
                & \lambda B_1
                \\ 
                -L - \lambda \tS_z^\top B_1^\top
                & \lambda \tR_z
                & 0
                & 0
                \\
                \begin{bmatrix}P & L^\top\end{bmatrix}^\top
                & 0
                & 0.5 \tau c_r^{-2} I_{N+1}
                & 0
                \\
                \lambda B_1^\top
                & 0
                & 0
                & - \lambda \tQ_z^{-1}
            \end{bmatrix}
            \succ 0
        \end{equation}
        \normalsize
        \medskip
        \vspace*{-0.6\baselineskip}
        \hrule
        \vspace*{-0.6\baselineskip}
    \end{figure*}
    and
    \begin{equation}\label{eq:stability-condition-LFR-invariance}
        \begin{bmatrix}
            P & P S_z & P & 0 \\
            S_z^\top P & \nu R_z & 0 & \nu \\
            P & 0 & -\nu Q_z^{-1} & 0 \\
            0 & \nu & 0 & 1
        \end{bmatrix}
        \succeq 0
    \end{equation}
    hold,
    then there exists an amount of data $d_0\in\bbN$ such that for all $d\geq d_0$ the controller $\mu(x) = L P^{-1} \hat{\Phi}(x)$ achieves exponential stability of the nonlinear system~\eqref{eq:dynamics-nonlinear} for all initial conditions $\hat{x}\in\cX_\mathrm{RoA} := \{x\in\bbR^n\,|\, \hat{\Phi}(x)^\top P^{-1} \hat{\Phi}(x) \leq 1\}$ with probability $1-\delta$.
\end{theorem}
\begin{proof}
    We divide the proof into two parts, where we first show that all $x\in\cX_\mathrm{RoA}$ satisfy $\hat{\Phi}(x)\in\mathbf{\Delta}_\Phi$ and then conclude positive invariance of $\cX_\mathrm{RoA}$ together with exponential stability of the closed-loop system~\eqref{eq:dynamics-nonlinear-feedback} for all $\hat{x}\in\cX_\mathrm{RoA}$.
    
    \textit{Part I: $x\in\cX_\mathrm{RoA}$ implies $\hat{\Phi}(x)\in\mathbf{\Delta}_\Phi$:~}
    The presented representation of the nonlinear system as an LFR~\eqref{eq:LFR-nonlinear-closed-loop-single-input} requires $\hat{\Phi}(x)\in\mathbf{\Delta}_\Phi$.
    Thus, first observe that~\eqref{eq:stability-condition-LFR-invariance} is equivalent to 
    \begin{equation*}
        \begin{bmatrix}
            \frac{1}{\nu}P & \frac{1}{\nu}PS_z \\
            \frac{1}{\nu}S_z^\top P & R_z
        \end{bmatrix}
        + \begin{bmatrix}
            \frac{1}{\nu}P \\ 0 
        \end{bmatrix}
        Q_z 
        \begin{bmatrix}
            \frac{1}{\nu}P \\ 0
        \end{bmatrix}^\top
        - \begin{bmatrix}
            0 \\ 1
        \end{bmatrix}
        \nu
        \begin{bmatrix}
            0 \\ 1
        \end{bmatrix}^\top
        \succeq 0,
    \end{equation*}
    where we first divide the inequality by $\nu$ and then apply the Schur complement (cf.~\cite{boyd:vandenberghe:2004}) twice.
    Multiplying from the left and from the right by $\diag(\nu P^{-1},1)$ yields
    \begin{equation*}
        \begin{bmatrix}
            Q_z & S_z \\ S_z^\top & R_z
        \end{bmatrix}
        - \nu \begin{bmatrix}
            - P^{-1} & 0 \\ 0 & 1
        \end{bmatrix}
        \succeq 0.
    \end{equation*}
    Multiplying from left and right by $
        \begin{bmatrix}
            \hat{\Phi}(x)^\top & 1
        \end{bmatrix}^\top
    $ and its transpose, respectively, where $x\in\cX_\mathrm{RoA}$, results in $\hat{\Phi}(x)\in\mathbf{\Delta}_\Phi$ for all $x\in\cX_\mathrm{RoA}$ (cf.\ the S-procedure~\cite{scherer:weiland:2000,boyd:vandenberghe:2004}), where $\mathbf{\Delta}_\Phi$ is defined in~\eqref{eq:condition-x-in-Z}. 

    \textit{Part II: Positive invariance of $\cX_\mathrm{RoA}$ and exponential stability:~}
    For positive invariance of $\cX_\mathrm{RoA}$, we have to show that $x(t+\chi)\in\cX_\mathrm{RoA}$ for all $x(t)\in\cX_\mathrm{RoA}$ and $\chi>0$. 
    According to the definition of $\cX_\mathrm{RoA}$, we define the Lyapunov function candidate $V(x) = \hat{\Phi}(x)^\top P^{-1} \hat{\Phi}(x)$ such that positive invariance follows if $\ddt{}V(x(t)) \leq 0$ for all trajectories satisfying $x(t)\in\cX_\mathrm{RoA}$, $t\geq0$.

    First, we define $K = L P^{-1}$ and recall $A_K = A + B_0 K$.
    Then, we apply the Schur complement to \eqref{eq:stability-condition-LFR-single-input} to obtain
    \begin{multline}\label{eq:proof-Schur-single-input}
        \begin{bmatrix}
            -A_K P - PA_K^\top
            -\tau I_N
            & \star
            & \star
            \\ 
            -KP - \lambda \tS_z^\top B_1^\top
            & \lambda \tR_z
            & \star
            \\
            -\begin{bmatrix}P\\KP\end{bmatrix}
            & 0
            & 0.5 \tau c_r^{-2} I_{N+1}
        \end{bmatrix}
        \\
        + \lambda^{-1}  
        \begin{bmatrix}
            \lambda B_1 \\ 0 \\ 0
        \end{bmatrix}
        \tQ_z
        \begin{bmatrix}
            \lambda B_1 \\ 0 \\ 0
        \end{bmatrix}^\top
        \succ 0.
    \end{multline}
    Note that 
    \begin{equation}\label{eq:proof-Schur-stability}
        \begin{bmatrix}
            -A_K P - PA_K^\top & \star & \star \\
            - KP & 0 & \star \\
            -\begin{bmatrix}P\\KP\end{bmatrix} & 0 & 0
        \end{bmatrix}
        \\
        = \begin{bmatrix}
            A_K & -I \\
            K & 0 \\
            \begin{bmatrix}I\\K\end{bmatrix} & 0
        \end{bmatrix}
        \begin{bmatrix}
            0 & P \\ P & 0
        \end{bmatrix}
        \begin{bmatrix}
            \star
        \end{bmatrix}^\top
    \end{equation}
    and 
    \begin{multline*}
        \begin{bmatrix}
            0 & - \lambda B_1\tS_z & 0 \\
            \star & \lambda \tR_z & 0 \\
            \star & \star & 0
        \end{bmatrix}
        + \lambda
        \begin{bmatrix}
            B_1 \\ 0 \\ 0
        \end{bmatrix}
        \tQ_z
        \begin{bmatrix}
            B_1 \\ 0 \\ 0
        \end{bmatrix}^\top
        \\
        = \lambda
        \begin{bmatrix}
            B_1 & 0 \\ 
            0 & -I \\
            0 & 0
        \end{bmatrix}
        \begin{bmatrix}
            \tQ_z & \tS_z \\ \tS_z^\top & \tR_z
        \end{bmatrix}
        \begin{bmatrix}
            B_1 & 0 \\ 
            0 & -I \\
            0 & 0
        \end{bmatrix}^\top.
    \end{multline*}
    Moreover, we define 
    \begin{equation}\label{eq:proof-multiplier-epsilon}
        \Pi_r = 
        \begin{bmatrix}
            -I_N & 0 \\ 0 & 2c_r^{2}I_{N+m}
        \end{bmatrix},
        \;
        \Pi_r^{-1} = 
        \begin{bmatrix}
            -I_N & 0 \\ 0 & 0.5c_r^{-2}I_{N+m}
        \end{bmatrix},        
    \end{equation}
    to obtain
    \begin{equation}\label{eq:proof-Schur-error}
        \begin{bmatrix}
            -\tau I_N & 0 & 0 \\
            0 & 0 & 0 \\
            0 & 0 & 0.5\tau c_r^{-2} I_{N+1}
        \end{bmatrix}
        = \tau 
        \begin{bmatrix}
            I & 0 \\ 0 & 0 \\ 0 & -I
        \end{bmatrix}
        \Pi_r^{-1}
        \begin{bmatrix}\star\end{bmatrix}^\top.
    \end{equation}
    Thus, we write~\eqref{eq:proof-Schur-single-input} equivalently as
    \begin{equation}
        \Psi^\top \diag\left(
            \begin{bmatrix} 
                0 & P \\ 
                P & 0
            \end{bmatrix}
            ,
            \lambda
            \begin{bmatrix}
                \tQ_z & \tS_z \\ \tS_z^\top & \tR_z
            \end{bmatrix}
            , 
            \tau
            \Pi_r^{-1}
        \right)
        \Psi 
        \succ 0,
    \end{equation}
    where 
    \begin{equation}
        \Psi^\top = 
        \left[\def\arraystretch{1.15}\begin{array}{cc|cc|cc}
            A_K & -I & B_1 & 0 & I & 0 \\
            K & 0 & 0 & -I & 0 & 0 \\
            \begin{bmatrix}I\\K\end{bmatrix} & 0 & 0 & 0 & 0 & -I
        \end{array}\right].
    \end{equation}
    Then, we apply the dualization lemma~\cite[Lm. 4.9]{scherer:weiland:2000} to arrive at
    \begin{equation}\label{eq:proof-primal-single-input}
        \tilde{\Psi}^\top \diag\left(
            \begin{bmatrix} 
                0 & \tP \\ 
                \tP & 0
            \end{bmatrix}
            ,
            \lambda^{-1}
            \begin{bmatrix}
                Q_z & S_z \\ S_z^\top & R_z
            \end{bmatrix}
            , 
            \tau^{-1}
            \Pi_r
        \right)
        \tilde{\Psi} 
        \succ 0,
    \end{equation}
    where $\tP=P^{-1}$ and
    \begin{equation}
        \tilde{\Psi}^\top = 
        \left[\def\arraystretch{1.15}\begin{array}{cc|cc|cc}
            I & A_K^\top & 0 & K^\top & 0 & \begin{bmatrix}I&K^\top\end{bmatrix} \\
            0 & B_1^\top & I & 0 & 0 & 0 \\
            0 & I & 0 & 0 & I & 0
        \end{array}\right]
    \end{equation}
    is constructed similarly to the discussion in~\cite[Sec. 8.1.2]{scherer:weiland:2000}.
    By multiplying~\eqref{eq:proof-primal-single-input} from the left and from the right by $
        \begin{bmatrix}
            \hat{\Phi}(x)^\top & \mu(x)\Delta_\Phi^\top & \varepsilon(\hat{\Phi}(x),\mu(x))^\top
        \end{bmatrix}^\top
    $ and its transpose, respectively, where $x\in\cX_\mathrm{RoA}$, $\mu(x)=K\hat{\Phi}(x)$, $\Delta_\Phi\in\mathbf{\Delta}_\Phi$, and $\varepsilon(\hat{\Phi}(x),\mu(x))$ satisfies~\eqref{eq:finite-gain-bound-epsilon}, we obtain
    \begin{multline}\label{eq:proof-S-procedure-single-input}
        \begin{bmatrix}\star\end{bmatrix}^\top 
        \begin{bmatrix}
            0 & \tP \\ \tP & 0
        \end{bmatrix}
        \begin{bmatrix}
            \hat{\Phi}(x) \\ 
            A_K \hat{\Phi}(x) + B_1 \Delta_\Phi\mu(x) + \varepsilon(\hat{\Phi}(x),\mu(x))
        \end{bmatrix}
        \\
        + \lambda^{-1} 
        \begin{bmatrix}\star\end{bmatrix}^\top 
        \begin{bmatrix}
            Q_z & S_z \\
            S_z^\top & R_z
        \end{bmatrix}
        \begin{bmatrix}
            \Delta_\Phi \mu(x) \\
            K \hat{\Phi}(x)
        \end{bmatrix}
        \\
        + \tau^{-1}
        \begin{bmatrix}\star\end{bmatrix}^\top 
        \Pi_r
        \begin{bmatrix}
            \varepsilon(\hat{\Phi}(x),\mu(x)) \\
            \begin{bmatrix}\hat{\Phi}(x) \\ K \hat{\Phi}(x)\end{bmatrix}
        \end{bmatrix}
        < 0.
    \end{multline}
   For $\Delta_\Phi\in\mathbf{\Delta}_\Phi$ and $\varepsilon(\hat{\Phi}(x),\mu(x))$ satisfying~\eqref{eq:finite-gain-bound-epsilon}, the last two summands of~\eqref{eq:proof-S-procedure-single-input} are nonnegative as
    \begin{equation*}
        \mu(x)^2
        \begin{bmatrix}
            \Delta_\Phi \\
            1
        \end{bmatrix}^\top
        \begin{bmatrix}
            Q_z & S_z \\
            S_z^\top & R_z
        \end{bmatrix}
        \begin{bmatrix}
            \Delta_\Phi \\
            1
        \end{bmatrix}
        \geq 0
    \end{equation*}
    with the definition of $\mathbf{\Delta}_\Phi$ in~\eqref{eq:condition-x-in-Z} and 
    \begin{align*}
        &\begin{bmatrix}\star\end{bmatrix}^\top 
        \Pi_r
        \begin{bmatrix}
            \varepsilon(\hat{\Phi}(x),\mu(x)) \\
            \begin{bmatrix}\hat{\Phi}(x) \\ K \hat{\Phi}(x)\end{bmatrix}
        \end{bmatrix}
        \nonumber\\
        &\quad
        = 2 c_r^2 (\|\hat{\Phi}(x)\|^2 + \|\mu(x)\|^2) -\|\varepsilon(\hat{\Phi}(x),\mu(x))\|^2
        \nonumber\\
        &\quad
        \geq c_r^2 (\|\hat{\Phi}(x)\| + \|\mu(x)\|)^2 -\|\varepsilon(\hat{\Phi}(x),\mu(x))\|^2
        \geq 0.
    \end{align*}
    Hence,~\eqref{eq:proof-S-procedure-single-input} implies
    \begin{multline}\label{eq:proof-nonlinear-dissipation-inequality-substitute-single-input}
        0 > \hat{\Phi}(x)^\top \tP (A_K \hat{\Phi}(x) + B_1 \Delta_\Phi\mu(x) + \varepsilon(\hat{\Phi}(x),\mu(x)))
        \\
        + (A_K \hat{\Phi}(x) + B_1 \Delta_\Phi \mu(x) + \varepsilon(\hat{\Phi}(x),\mu(x)))^\top \tP \hat{\Phi}(x)
    \end{multline}
    if $x\in\cX_\mathrm{RoA}\setminus\{0\}$, $\Delta_\Phi\in\mathbf{\Delta}_\Phi$, and if $\varepsilon(\hat{\Phi}(x),\mu(x))$ satisfies~\eqref{eq:finite-gain-bound-epsilon} (again, we refer to the S-procedure~\cite{scherer:weiland:2000,boyd:vandenberghe:2004}).
    
    Recall the Lyapunov candidate function $V(x) = \hat{\Phi}(x)^\top \tP \hat{\Phi}(x)$ which satisfies $V(0)=0$ and
    \begin{subequations}\label{eq:proof-Lyapunov-upper-lower-bounds}
        \begin{alignat}{2}
            V(x) &\leq \sigma_{\max}(\tP) \|\hat{\Phi}(x)\|^2& &\leq \sigma_{\max}(\tP) L_\Phi^2\|x\|^2,
            \\
            V(x) &\geq \sigma_{\min}(\tP) \|\hat{\Phi}(x)\|^2& &\geq \sigma_{\min}(\tP) \|x\|^2
        \end{alignat}        
    \end{subequations}
    due to $\|\hat{\Phi}(x)\| = \|\Phi(x)-\Phi(0)\|$,~\eqref{eq:philower}, and~\eqref{eq:phiupper}.
    Further, substituting $\ddt{}\hat{\Phi}(x(t)) = A_K \hat{\Phi}(x(t)) + B_1 \Delta_\Phi(t)\mu(x(t)) + \varepsilon(\hat{\Phi}(x(t)),\mu(x(t)))$ in~\eqref{eq:proof-nonlinear-dissipation-inequality-substitute-single-input} leads to 
    \begin{multline}\label{eq:proof-Lyapunov-decay-single-input}
        \tfrac{\mathrm{d}}{\mathrm{d}t}V(x(t)) 
        = \hat{\Phi}(x(t))^\top \tP \bigl(\tfrac{\mathrm{d}}{\mathrm{d}t}\hat{\Phi}(x(t))\bigl)
        \\
        + \bigl(\tfrac{\mathrm{d}}{\mathrm{d}t}\hat{\Phi}(x(t))\bigl)^\top \tP \hat{\Phi}(x(t)) < 0
    \end{multline}
    if $x(t)\in\cX_\mathrm{RoA}\setminus\{0\}$, $\Delta_\Phi(t)\in\mathbf{\Delta}_\Phi$, and if $\varepsilon(\hat{\Phi}(x),\mu(x))$ satisfies~\eqref{eq:finite-gain-bound-epsilon}.
    Hence, the controller $\mu(x(t))=K\hat{\Phi}(x(t))$ guarantees the decrease~\eqref{eq:proof-Lyapunov-decay-single-input} for solutions of the LFR~\eqref{eq:LFR-nonlinear-closed-loop-single-input} and, thus, also for the lifted system~\eqref{eq:dynamics-lifted-single-input-Delta} if $x(t)\in\cX_\mathrm{RoA}\setminus\{0\}$, $\Delta_\Phi(t)\in\mathbf{\Delta}_\Phi$, and if $\varepsilon(\hat{\Phi}(x(t)),\mu(x(t)))$ satisfies~\eqref{eq:finite-gain-bound-epsilon}.
    Recall that $\hat{\Phi}(x(t))\in\mathbf{\Delta}_\Phi$ for all $x(t)\in\cX_\mathrm{RoA}$ according to Part I and $\varepsilon(\hat{\Phi}(x(t)),\mu(x(t)))=\hat{r}(x(t),u(t))$ due to~\eqref{eq:epsilon-definition}.
    Thus, the controller $\mu(x(t))=K\hat{\Phi}(x(t))$ guarantees~\eqref{eq:proof-Lyapunov-decay-single-input} in particular for system~\eqref{eq:dynamics-lifted-single-input} for all $x(t)\in\cX_\mathrm{RoA}\setminus\{0\}$ if $\hat{r}(x(t),\mu(x(t)))$ satisfies~\eqref{eq:finite-gain-bound-remainder}.
    This concludes positive invariance of $\cX_\mathrm{RoA}$ and, thus, robust exponential stability of~\eqref{eq:dynamics-lifted-single-input} for all $x(t)\in \cX_\mathrm{RoA}$ if $\hat{r}(x(t),\mu(x(t)))$ satisfies~\eqref{eq:finite-gain-bound-remainder}. 
    Finally, we exploit Proposition~\ref{prop:sampling-error-bound-Koopman-generator} and Lemma~\ref{lm:dynamics-lifted} to deduce that the lifted system based on the data-driven approximation $\cL^u_d$ of the Koopman generator and its remainder $\hat{r}(x(t),\mu(x(t)))$ for a data length $d\geq d_0$ satisfies the obtained bound with probability $1-\delta$. 
    Hence, we conclude exponential stability of the nonlinear system~\eqref{eq:dynamics-nonlinear} for all $x(t)\in\cX_\mathrm{RoA}$ with probability $1-\delta$.
\end{proof}
Theorem~\ref{thm:stability-condition-LFR-single-input} establishes a robust controller design w.r.t.\ two uncertainties, namely the uncertain remainder $\hat{r}(x,u)$ and the (artificial) uncertainty $\Phi(x)$ which are both bounded.
We emphasize that the controller design relies on an SDP in terms of LMIs which can be efficiently solved.
The design of this robust controller is based on the sampling error bounds derived in~\cite{schaller:worthmann:philipp:peitz:nuske:2023}.
In particular, the contributions of the presented controller design are the obtained closed-loop guarantees for the unknown nonlinear system~\eqref{eq:dynamics-nonlinear} based on data.
We note that our controller design exploits the fact that the control input also enters state-independently into the lifted dynamics via $B_0$. 
More precisely,~\eqref{eq:stability-condition-LFR-single-input} can only be feasible if $(A,B_0)$ is stabilizable.

A key technical ingredient of the proof is to derive Lyapunov stability directly for the \emph{true} nonlinear system with state $x$ instead of restricting it to the lifted bilinear system with state $\hat{\Phi}(x)$.
Thus, the considered error bound based on the Koopman approximation can be directly used when proving closed-loop stability of the nonlinear system~\eqref{eq:dynamics-nonlinear} without additional consideration of the resulting submanifold $\operatorname{im}(\hat{\Phi})$ resulting by the chosen lifting, cf.~\cite{goor:mahony:schaller:worthmann:2023}.
In addition, the tutorial-style presentation and proof of Theorem~\ref{thm:stability-condition-LFR-single-input} contribute to a comprehensive controller design that can be widely used in various disciplines.
Even though we only obtain probabilistic guarantees in Theorem~\ref{thm:stability-condition-LFR-single-input}, we exploit the direct relation between the necessary amount of data $d_0$ for a given desired accuracy of the lifting and certainty about the closed-loop guarantees presented in~\eqref{eq:mindata}. 
Hence, the proposed controller design yields end-to-end guarantees for the underlying true nonlinear system~\eqref{eq:dynamics-nonlinear} based on measured data.
Note that the resulting RoA is a sublevel set of the Lyapunov function and can always be scaled such that $\cX_\mathrm{RoA}\subseteq \bbX$.

The proposed controller design is summarized in Algorithm~\ref{alg:controller-design-single-input}. 
Overall, we combine the EDMD estimate of the Koopman generator in~\eqref{eq:EDMD-optimization} and the error bounds of Proposition~\ref{prop:sampling-error-bound-Koopman-generator} with the robust controller design in Theorem~\ref{thm:stability-condition-LFR-single-input} leading to closed-loop guarantees of the nonlinear system~\eqref{eq:dynamics-nonlinear}.
The Koopman-based surrogate model obtained here only requires solving the least-squares optimization problem~\eqref{eq:EDMD-optimization}, which scales well to higher-dimensional systems. 
The controller design approach, on the other hand, relies on solving an SDP with complexity $\cO(N^6)$. 
Thus, future work should be devoted to including structure-exploiting SDP techniques~\cite{deklerk:2010,gramlich:holicki:scherer:ebenbauer:2023}. 

\begin{algorithm}[tb]
    \caption{State-feedback controller design for single inputs corresponding to Theorem~\ref{thm:stability-condition-LFR-single-input}.}\label{alg:controller-design-single-input}
    \textbf{Input:}
    \begin{itemize}
        \item Data $\{x_j^{\bar{u}}, \dot{x}_j^{\bar{u}}\}_{j=1}^{d^{\bar{u}}}$ for ${\bar{u}}\in\{0,1\}$, where $d^{\bar{u}}\geq d_0$ for $d_0$ according to~\eqref{eq:mindata}
        \item Lifting $\Phi(x) = \begin{bmatrix}1 & \hat{\Phi}(x)\end{bmatrix}^\top$ defined in \eqref{eq:lifting-function} 
        \item Probabilistic tolerance $\delta\in(0,1)$ and error bound $c_r > 0$
        \item Ellipsoidal region $\mathbf{\Delta}_\Phi$ defined by $Q_z\prec 0$, $S_z$, $R_z > 0$
    \end{itemize}
    \textbf{Output:} Controller $u=\mu(x)$ for the nonlinear system~\eqref{eq:dynamics-nonlinear} which is exponentially stabilizing with probability $1-\delta$ if $\hat{x}\in\cX_\mathrm{RoA}$

    \emph{Data-driven system representation:}
    \begin{algorithmic}
        \State Arrange the data according to~\eqref{eq:data-matrices} in $X^{\bar{u}}, Y^{\bar{u}}$
        \State Solve the optimization problem~\eqref{eq:EDMD-optimization} to obtain the data-based system matrices $A$, $B_0$, $B_1$ according to Lemma~\ref{lm:dynamics-lifted}
    \end{algorithmic}
    
    \emph{Controller design:}
    \begin{algorithmic}
        \State Solve the LMI feasibility problem in Theorem~\ref{thm:stability-condition-LFR-single-input}
        \If{successful}
            \State Define the controller $\mu(x) = LP^{-1}\hat{\Phi}(x)$
        \EndIf
    \end{algorithmic}
\end{algorithm}
Theorem~\ref{thm:stability-condition-LFR-single-input} allows to compute state-feedback laws of the form $\mu(x)=K\hat{\Phi}(x)$ that stabilize nonlinear single-input systems. 
In Section~\ref{sec:controller-design-general} (Theorem~\ref{thm:stability-condition-LFR}), we generalize this result into multiple directions.
In particular, we exploit that the uncertainty signal $w_{\Phi}$ in~\eqref{eq:LFR-nonlinear-closed-loop-single-input} is, in fact, measurable during closed-loop operation given that $\Delta_\Phi=\hat{\Phi}(x)$ is known.
This allows us to parametrize the controller as $\mu(x)=K\hat{\Phi}(x) + K_w w_\Phi$ for some matrix $K_w$. 
In contrast to Theorem~\ref{thm:stability-condition-LFR-single-input}, this yields a state feedback which depends nonlinearly on $\hat{\Phi}(x)$ as $\mu(x) = (I-K_w\hat{\Phi}(x))^{-1}K\hat{\Phi}(x)$.
The resulting more flexible controller structure provides less conservative stability conditions and, in general, a larger RoA $\cX_\mathrm{RoA}$ than in Theorem~\ref{thm:stability-condition-LFR-single-input}.
Further, we address the multi-input case.
\section{General state-feedback laws reducing conservatism}\label{sec:controller-design-general}
After presenting the main idea of the proposed controller design, we now discuss the design for general input dimensions ($m\geq 1$) and a more flexible controller structure reducing conservatism compared to the design in Section~\ref{sec:controller-design-simple}. 
To this end, we recall the lifted dynamics in~\eqref{eq:dynamics-lifted} and define $\tB=\begin{bmatrix}B_1 & \cdots & B_m\end{bmatrix}$ to compactly write the dynamics as 
\begin{equation}\label{eq:dynamics-lifted-Kronecker}
    \ddt{}\hat{\Phi}(x) = A \hat{\Phi}(x) + B_0 u + \tB (u\kron \hat{\Phi}(x)) + \hat{r}(x,u).
\end{equation}
Again, we impose that $\hat{\Phi}(x)\in\mathbf{\Delta}_\Phi$ with $\mathbf{\Delta}_\Phi$ defined prior to \eqref{eq:condition-x-in-Z} and that the remainder $\hat{r}(x,u)$ satisfies~\eqref{eq:finite-gain-bound-remainder} for all $x\in\bbX$ and $u\in\bbU$.
Then, the lifted system~\eqref{eq:dynamics-lifted-Kronecker} can be expressed as the open-loop LFR
\begin{subequations}\label{eq:LFR-nonlinear-open-loop}
    \begin{align}
        \begin{bmatrix}
            \ddt{}\hat{\Phi}(x) \\ u \\ \begin{bmatrix}v_1\\v_2\end{bmatrix}
        \end{bmatrix}
        &= \begin{bmatrix}
            A & B_0 & \tB & I \\
            0 & I & 0 & 0 \\
            \begin{bmatrix} I \\ 0\end{bmatrix} & \begin{bmatrix}0\\I\end{bmatrix} & 0 & 0
        \end{bmatrix}
        \begin{bmatrix}
            \hat{\Phi}(x) \\ u \\ w_\Phi \\ w_r
        \end{bmatrix},\\
        w_\Phi &= (I_m\kron \Delta_\Phi) u, \label{eq:LFR-nonlinear-open-loop:uncertainty-bilinearity}\\
        w_r &= \varepsilon(v_1,v_2) \label{eq:LFR-nonlinear-open-loop:uncertainty-remainder}
    \end{align}
\end{subequations}
with $\Delta_\Phi\in\mathbf{\Delta}_\Phi$ and $\varepsilon$ as defined in~\eqref{eq:epsilon-definition} satisfying the bound~\eqref{eq:finite-gain-bound-epsilon} for all $x\in\bbX,u\in\bbU$.
As main difference to~\eqref{eq:LFR-nonlinear-open-loop-single-input}, the LFR~\eqref{eq:LFR-nonlinear-open-loop} is exposed to the uncertainty $(I_m\kron \Delta_\Phi)$ instead of $\Delta_\Phi$. 
To derive an uncertainty characterization of $(I_m\kron \Delta_\Phi)$, we define the set 
\begin{equation}\label{eq:Delta-representation-Delta-structured}
    \mathbf{\Delta} 
    \coloneqq 
    \left\{
        \Delta \in \bbR^{mN\times m}
    \middle|
        \begin{bmatrix}
            \Delta \\ I
        \end{bmatrix}^\top 
        \Pi_{\Delta}
        \begin{bmatrix}
            \Delta \\ I
        \end{bmatrix}
        \succeq 0
        \;\forall\,
        \Pi_{\Delta} \in \mathbf{\Pi}_{\Delta}
    \right\}
\end{equation}
for some multiplier class $\mathbf{\Pi}_{\Delta}$ defined in the following as a suitable convex cone of symmetric matrices.
More precisely, $\mathbf{\Pi}_\Delta$ needs to be chosen such that $(I_m\kron \Delta_\Phi)\in\mathbf{\Delta}$ for all $\Delta_\Phi\in\mathbf{\Delta}_\Phi$.
Thus, we parameterize the multiplier class via the LMI representation
\begin{equation}\label{eq:Delta-representation-Pi-structured}
    \mathbf{\Pi}_{\Delta} \coloneqq 
    \left\{
        \Pi_{\Delta} 
    \middle|
        \Pi_{\Delta} = \begin{bmatrix}
            \tilde{\Lambda} \kron Q_z & \tilde{\Lambda} \kron S_z \\
            \tilde{\Lambda} \kron S_z^\top & \tilde{\Lambda} \kron R_z
        \end{bmatrix}, 0\preceq\tilde{\Lambda}\in\bbR^{m\times m}
    \right\}.
\end{equation}
Using~\cite[Prop. 2]{strasser:berberich:allgower:2023b}, we deduce that $\Delta\in\mathbf{\Delta}$ if and only if $\Delta=I_m\kron\Delta_\Phi$ with $\Delta_\Phi\in\mathbf{\Delta}_\Phi$. 
In particular, $\mathbf{\Delta}$ with multiplier class $\mathbf{\Pi}_{\Delta}$ exploits the structure of the uncertainty $\Delta=I_m\kron\Delta_\Phi$ without additional conservatism. 
For scalar inputs, i.e., $m=1$, the uncertainty description $\mathbf{\Delta}$ reduces to $\mathbf{\Delta}_\Phi$, i.e., $\mathbf{\Delta}=\mathbf{\Delta}_\Phi$.

In the following, we exploit that we can access the uncertainty $\Delta_\Phi=\hat{\Phi}(x)$ and, hence, $\Delta=I_m\kron\hat{\Phi}(x)$ in the LFR channel~\eqref{eq:LFR-nonlinear-open-loop:uncertainty-bilinearity} by employing gain-scheduling techniques~\cite{scherer:2001} to design a flexible state-feedback controller exponentially stabilizing the closed-loop.
In particular, we design a full-information feedback controller~\cite{doyle:glover:khargonekar:francis:1989,packard:zhou:pandey:leonhardson:balas:1992,astolfi:1997} that depends not only on the lifted state $\hat{\Phi}(x)$ but also on the bilinearity $w_\Phi=\Delta u$.
More precisely, we consider the feedback parametrization
\begin{equation}\label{eq:controller-nonlinear} 
    u = \mu(x) = K\hat{\Phi}(x) + K_w w_\Phi,
\end{equation}
where $K\in\bbR^{m\times N}$ and $K_w\in\bbR^{m\times Nm}$. 
Again, we abbreviate $A_K = A+B_0K$ and set $B_{K_w} = \tB + B_0K_w$.
Substituting \eqref{eq:controller-nonlinear} 
in the open-loop LFR~\eqref{eq:LFR-nonlinear-open-loop}, we obtain the corresponding closed-loop LFR
\begin{subequations}\label{eq:LFR-nonlinear-closed-loop}
    \begin{align}
        \begin{bmatrix}
            \ddt{}\hat{\Phi}(x) \\ \mu(x) \\ \begin{bmatrix}v_1\\v_2\end{bmatrix}
        \end{bmatrix}
        &= \begin{bmatrix}
            A_K & B_{K_w} & I \\
            K & K_w & 0 \\
            \begin{bmatrix} I \\ K \end{bmatrix} & \begin{bmatrix}0\\K_w\end{bmatrix} & 0
        \end{bmatrix}
        \begin{bmatrix}
            \hat{\Phi}(x) \\ w_\Phi \\ w_r
        \end{bmatrix},\\
        w_\Phi &= \Delta \mu(x), \\
        w_r &= \varepsilon(v_1,v_2)
    \end{align}
\end{subequations}%
with $\Delta\in\mathbf{\Delta}$ and where $\varepsilon$ satisfies~\eqref{eq:finite-gain-bound-epsilon} for all $x\in\bbX,u\in\bbU$.
Note that the LFR~\eqref{eq:LFR-nonlinear-closed-loop} reduces to~\eqref{eq:LFR-nonlinear-closed-loop-single-input} in the case $m=1$ and $K_w=0$.

The following theorem establishes a controller design method guaranteeing exponential stability of the nonlinear system~\eqref{eq:dynamics-nonlinear} based on robust exponential stability of the LFR~\eqref{eq:LFR-nonlinear-closed-loop}. 
\begin{theorem}\label{thm:stability-condition-LFR}
    Let Assumption~\ref{ass:invariance-of-dictionary} hold. 
    Suppose a desired error bound $c_r>0$ and a probabilistic tolerance $\delta \in (0,1)$ in the sense of Proposition~\ref{lm:dynamics-lifted} are given.
    If there exists a matrix $0\prec P=P^\top\in\bbR^{N\times N}$, matrices $L\in\bbR^{m\times N}$, $L_w\in\bbR^{m\times Nm}$, a matrix $0\prec \Lambda=\Lambda^\top\in\bbR^{m\times m}$, and scalars $\nu>0$, $\tau>0$ such that~\eqref{eq:stability-condition-LFR}
    \begin{figure*}[!t]
        \small
        \begin{equation}\label{eq:stability-condition-LFR}
            \begin{bmatrix}
                -AP - B_0L - PA^\top - L^\top B_0^\top - \tau I_N
                & \star 
                & \star
                & \star
                \\ 
                -L - (\Lambda \kron \tS_z^\top)\tB^\top
                - (I_m\kron\tS_z^\top)L_w^\top B_0^\top
                & \Lambda \kron \tR_z 
                - L_w (I_m\kron\tS_z) - (I_m\kron\tS_z^\top) L_w^\top
                & \star
                & \star
                \\
                -\begin{bmatrix}P & L^\top\end{bmatrix}^\top
                & -\begin{bmatrix}0&L_w^\top\end{bmatrix}^\top(I_m\kron\tS_z)
                & \tau0.5c_r^{-2}I_{N+m}
                & \star
                \\
                (\Lambda\kron I_N)\tB^\top + L_w^\top B_0^\top
                & L_w^\top
                & -\begin{bmatrix}0&L_w^\top\end{bmatrix}
                & - \Lambda \kron \tQ_z^{-1}
            \end{bmatrix}\succ 0
        \end{equation}
        \normalsize
        \medskip
        \vspace*{-0.6\baselineskip}
        \hrule
        \vspace*{-0.6\baselineskip}
    \end{figure*}
    and
    \begin{equation}
        \begin{bmatrix}
            P & P S_z & P & 0 \\
            S_z^\top P & \nu R_z & 0 & \nu \\
            P & 0 & -\nu Q_z^{-1} & 0 \\
            0 & \nu & 0 & 1
        \end{bmatrix}
        \succeq 0
    \end{equation}
    hold,
    then there exists an amount of data $d_0\in\bbN$ such that for all $d\geq d_0$ the controller
    \begin{equation}\label{eq:controller-nonlinear-explicit}
        \mu(x) = (I-L_w(\Lambda^{-1}\kron \hat{\Phi}(x)))^{-1} L P^{-1} \hat{\Phi}(x)
    \end{equation}
    achieves exponential stability of the nonlinear system~\eqref{eq:dynamics-nonlinear} for all initial conditions $\hat{x}\in\cX_\mathrm{RoA} = \{x\in\bbR^n\,|\, \hat{\Phi}(x)^\top P^{-1} \hat{\Phi}(x) \leq 1\}$ with probability $1-\delta$.
\end{theorem}
\begin{proof}
    As in Theorem~\ref{thm:stability-condition-LFR-single-input}, we divide the proof into two parts. We first show that all $x\in\cX_\mathrm{RoA}$ satisfy $(I_m\kron\hat{\Phi}(x))\in\mathbf{\Delta}$ and then conclude positive invariance of~$\cX_\mathrm{RoA}$ together with exponential stability of~\eqref{eq:dynamics-nonlinear} for all $\hat{x}\in\cX_\mathrm{RoA}$.
    
    \textit{Part I: $x\in\cX_\mathrm{RoA}$ implies $(I_m\kron \hat{\Phi}(x))\in\mathbf{\Delta}$:~}
    It can be shown analogously to the proof of Theorem~\ref{thm:stability-condition-LFR-single-input} that all $x\in\cX_\mathrm{RoA}$ satisfy $\hat{\Phi}(x)\in\mathbf{\Delta}_\Phi$.
    More precisely,~\cite[Prop.~2]{strasser:berberich:allgower:2023b} implies that $\Delta\in\mathbf{\Delta}$ if and only if $\Delta=I_m\kron\Delta_\Phi$ with $\Delta_\Phi\in\mathbf{\Delta}_\Phi$.
    The claim then follows as $x\in \cX_\mathrm{RoA}$ implies $\hat{\Phi}(x)\in \mathbf{\Delta}_\Phi$.
    
    \textit{Part II: Positive invariance of $\cX_\mathrm{RoA}$ and exponential stability:~}
    For positive invariance of $\cX_\mathrm{RoA}$ together with exponential stability of~\eqref{eq:dynamics-nonlinear} for all $\hat{x}\in\cX_\mathrm{RoA}$, we define analogously to the proof of Theorem~\ref{thm:stability-condition-LFR-single-input} the Lyapunov function candidate $V(x)=\hat{\Phi}(x)^\top P^{-1}\hat{\Phi}(x)$ such that positive invariance of $\cX_\mathrm{RoA}$ can be deduced if $\ddt{}V(x(t))\leq 0$ for all $x(t)\in\cX_\mathrm{RoA}$.

    First, we define $K = L P^{-1}$ and $K_w = L_w(\Lambda^{-1} \kron I_N)$, and denote $A_K=A+B_0K$ and $B_{K_w}=\tB+B_0K_w$. 
    Then, we apply the Schur complement to~\eqref{eq:stability-condition-LFR} to obtain 
    \small
    \begin{multline}\label{eq:proof-Schur}
        \hspace*{-10pt}
        \begin{bmatrix}
            -A_K P - PA_K^\top
            -\tau I_N
            & \star
            & \star
            \\ 
            -KP - (\Lambda\kron\tS_z^\top) B_{K_w}^\top
            & H_{K_w}
            & \star
            \\
            -\begin{bmatrix}P\\KP\end{bmatrix}
            & -\begin{bmatrix}0\\K_w\end{bmatrix}(\Lambda\kron\tS_z)
            & 0.5 \tau c_r^{-2} I_{N+1}
        \end{bmatrix}
        \\
        +  
        \begin{bmatrix}
            B_{K_w}(\Lambda\kron I_N) \\ K_w(\Lambda\kron I_N) \\ -\begin{bmatrix}0\\K_w(\Lambda\kron I_N)\end{bmatrix}
        \end{bmatrix}
        (\Lambda^{-1}\kron\tQ_z)
        \begin{bmatrix}
            \star
        \end{bmatrix}^\top
        \succ 0,
    \end{multline}
    \normalsize
    where $H_{K_w}=(\Lambda\kron \tR_z) - K_w(\Lambda\kron\tS_z) - (\Lambda\kron\tS_z^\top)K_w^\top$.
    Recall~\eqref{eq:proof-Schur-stability},~\eqref{eq:proof-Schur-error} and note that
    \begin{multline*}
        \begin{bmatrix}
            0 & - B_{K_w}(\Lambda\kron\tS_z) & 0 \\
            \star & H_{K_w} & \star \\
            \star & -\begin{bmatrix}0\\K_w\end{bmatrix}(\Lambda\kron\tS_z) & 0
        \end{bmatrix}
        +  
        \begin{bmatrix}
            B_{K_w} \\ K_w \\ -\begin{bmatrix}0\\K_w\end{bmatrix}
        \end{bmatrix}
        (\Lambda\kron\tQ_z)
        \begin{bmatrix}
            \star
        \end{bmatrix}^\top
        \\
        =
        \begin{bmatrix}
            B_{K_w} & 0 \\ 
            K_w & -I \\
            \begin{bmatrix}0\\K_w\end{bmatrix} & 0
        \end{bmatrix}
        \begin{bmatrix}
            \Lambda \kron \tQ_z & \Lambda \kron \tS_z \\
            \Lambda \kron \tS_z^\top & \Lambda \kron \tR_z
        \end{bmatrix}
        \begin{bmatrix}
            \star
        \end{bmatrix}^\top.
    \end{multline*}
    Moreover, we observe 
    \begin{equation*}
        \Pi_{\Delta} 
        = T \left(
            \tilde{\Lambda} \kron \begin{bmatrix} Q_z & S_z \\ S_z^\top & R_z \end{bmatrix}
        \right) T^\top,
    \end{equation*}
    where 
    \begin{equation}\label{eq:proof-T-unitary}
        T = \begin{bmatrix}
            I_m \kron \begin{bmatrix}
                I_N & 0_{N\times 1}
            \end{bmatrix} \\
            I_m \kron \begin{bmatrix}
                0_{1\times N} & 1
            \end{bmatrix}
        \end{bmatrix}
    \end{equation}
    is orthogonal, i.e., $T^{-1}=T^\top$. Then, we directly obtain 
    \begin{equation*}
        \Pi_{\Delta}^{-1} 
        = T \left(
            \tilde{\Lambda}^{-1} \kron \begin{bmatrix} \tQ_z & \tS_z \\ \tS_z^\top & \tR_z \end{bmatrix}
        \right) T^\top
    \end{equation*}
    using $(V\kron W)^{-1} = V^{-1} \kron W^{-1}$. By the definition of $T$ in~\eqref{eq:proof-T-unitary} we deduce that 
    \begin{equation}\label{eq:multiplier-Delta-inverse}
        \Pi_{\Delta}^{-1}
        = \begin{bmatrix}
            \Lambda \kron \tQ_z & \Lambda \kron \tS_z \\
            \Lambda \kron \tS_z^\top & \Lambda \kron \tR_z
        \end{bmatrix}
    \end{equation}
    for $\tilde{\Lambda} = \Lambda^{-1}$.
    Thus, we write~\eqref{eq:proof-Schur} equivalently as
    \begin{equation*}
        \Psi^\top \diag\left(
            \begin{bmatrix} 
                0 & P \\ 
                P & 0
            \end{bmatrix}
            ,
            \Pi_\Delta^{-1}
            , 
            \tau
            \Pi_r^{-1}
        \right)
        \Psi 
        \succ 0,
    \end{equation*}
    where we recall $\Pi_r$ in~\eqref{eq:proof-multiplier-epsilon} and define
    \begin{equation*}
        \Psi^\top = 
        \left[\def\arraystretch{1.15}\begin{array}{cc|cc|cc}
            A_K & -I & B_{K_w} & 0 & I & 0 \\
            K & 0 & K_w & -I & 0 & 0 \\
            \begin{bmatrix}I\\K\end{bmatrix} & 0 & \begin{bmatrix}0\\K_w\end{bmatrix} & 0 & 0 & -I
        \end{array}\right].
    \end{equation*}
    Then, we apply the dualization lemma~\cite[Lm. 4.9]{scherer:weiland:2000} to arrive at 
    \begin{equation}\label{eq:proof-primal}
        \tilde{\Psi}^\top \diag\left(
            \begin{bmatrix} 
                0 & \tP \\ 
                \tP & 0
            \end{bmatrix}
            ,
            \Pi_\Delta
            , 
            \tau^{-1}
            \Pi_r
        \right)
        \tilde{\Psi} 
        \succ 0,
    \end{equation}
    where $\tP = P^{-1}$ and
    \begin{equation*}
        \tilde{\Psi}^\top = 
        \left[\def\arraystretch{1.15}\begin{array}{cc|cc|cc}
            I & A_K^\top & 0 & K^\top & 0 & \begin{bmatrix}I&K^\top\end{bmatrix} \\
            0 & B_{K_w}^\top & I & K_w^\top & 0 & \begin{bmatrix}0&K_w^\top\end{bmatrix} \\
            0 & I & 0 & 0 & I & 0
        \end{array}\right]
    \end{equation*}
    along the lines of~\cite[Sec. 8.1.2]{scherer:weiland:2000}.
    By multiplying~\eqref{eq:proof-primal} from left and right by $
        \begin{bmatrix}
            \hat{\Phi}(x)^\top & (\Delta \mu(x))^\top & \varepsilon(\hat{\Phi}(x),\mu(x))^\top
        \end{bmatrix}^\top
    $ and its transpose, respectively, where $x\in\cX_\mathrm{RoA}$, $\mu(x)$ as in~\eqref{eq:controller-nonlinear-explicit}, $\Delta\in\mathbf{\Delta}$, and $\varepsilon(\hat{\Phi}(x),\mu(x))$ satisfies~\eqref{eq:finite-gain-bound-epsilon}, we obtain
    \begin{align}
        0 >& \begin{bmatrix}\star\end{bmatrix}^\top 
        \begin{bmatrix}
            0 & \tP \\ \tP & 0
        \end{bmatrix}
        \begin{bmatrix}
            \hat{\Phi}(x) \\ 
            A_K \hat{\Phi}(x) + B_{K_w} \Delta \mu(x) + \varepsilon(\hat{\Phi}(x),\mu(x))
        \end{bmatrix}
    \nonumber\\
        &+ \begin{bmatrix}\star\end{bmatrix}^\top 
        \Pi_\Delta
        \begin{bmatrix}
            \Delta \mu(x) \\
            K\hat{\Phi}(x) + K_w \Delta \mu(x)
        \end{bmatrix}
    \nonumber\\
        &+ \tau^{-1}\begin{bmatrix}\star\end{bmatrix}^\top 
        \Pi_r
        \begin{bmatrix}
            \varepsilon(\hat{\Phi}(x),\mu(x)) \\
            \begin{bmatrix}
                \hat{\Phi}(x) \\
                K\hat{\Phi}(x) + K_w \Delta \mu(x)
            \end{bmatrix}
        \end{bmatrix}.
    \label{eq:proof-S-procedure}
    \end{align}
    The last two summands of~\eqref{eq:proof-S-procedure} are nonnegative for $\Delta\in\mathbf{\Delta}$ and $\varepsilon(\hat{\Phi}(x),\mu(x))$ satisfying~\eqref{eq:finite-gain-bound-epsilon} as
    \begin{equation}\label{eq:proof-inequality-bilinearity-u}
        \mu(x)^\top
        \begin{bmatrix}
            \Delta \\
            I_m
        \end{bmatrix}^\top 
        \Pi_\Delta
        \begin{bmatrix}
            \Delta \\
            I_m
        \end{bmatrix}
        \mu(x)
        \geq 0
    \end{equation}
    with the definition of $\mathbf{\Delta}$ in~\eqref{eq:Delta-representation-Delta-structured} and 
    \begin{multline*}
        \begin{bmatrix}\star\end{bmatrix}^\top 
        \Pi_r
        \begin{bmatrix}
            \varepsilon(\hat{\Phi}(x),\mu(x)) \\
            \begin{bmatrix}
                \hat{\Phi}(x) \\ 
                K\hat{\Phi}(x) + K_w \Delta \mu(x)
            \end{bmatrix}
        \end{bmatrix}
        \\
        \geq c_r^2 (\|\hat{\Phi}(x)\| + \|\mu(x)\|)^2 -\|\varepsilon(\hat{\Phi}(x),\mu(x))\|^2.
    \end{multline*}
    Hence,~\eqref{eq:proof-S-procedure} implies
    \begin{multline}\label{eq:proof-nonlinear-dissipation-inequality-substitute}
        0 > \hat{\Phi}(x)^\top \tP (A_K \hat{\Phi}(x) + B_{K_w}\Delta \mu(x) + \varepsilon(\hat{\Phi}(x),\mu(x))) 
        \\
        + (A_K \hat{\Phi}(x) + B_{K_w}\Delta \mu(x) + \varepsilon(\hat{\Phi}(x),\mu(x)))^\top \tP \hat{\Phi}(x)
    \end{multline}
    if $x\in\cX_\mathrm{RoA}\setminus\{0\}$, $\Delta\in\mathbf{\Delta}$, and $\varepsilon(\hat{\Phi}(x),\mu(x))$ satisfies~\eqref{eq:finite-gain-bound-epsilon} (cf.~ the S-procedure~\cite{scherer:weiland:2000,boyd:vandenberghe:2004}).
    Recall the Lyapunov candidate function $V(x)$ with $\ddt{}\hat{\Phi}(x(t)) = A_K \hat{\Phi}(x(t)) + B_{K_w} \Delta(t)\mu(x(t)) + \varepsilon(\hat{\Phi}(x(t)),\mu(x(t)))$, $(I_m\kron\hat{\Phi}(x(t)))\in\mathbf{\Delta}$, and $\varepsilon(\hat{\Phi}(x(t)),\mu(x(t)))$ satisfying~\eqref{eq:finite-gain-bound-epsilon} if $\hat{r}(x(t),\mu(x(t)))$ satisfies~\eqref{eq:finite-gain-bound-remainder} for all $x(t)\in\cX_\mathrm{RoA}$.
    Then, positive invariance of $\cX_\mathrm{RoA}$ and, thus, exponential stability of the nonlinear system~\eqref{eq:dynamics-nonlinear} for all $x(t)\in\cX_\mathrm{RoA}$ with probability $1-\delta$ can be shown analogously to the proof of Theorem~\ref{thm:stability-condition-LFR-single-input} invoking Proposition~\ref{prop:sampling-error-bound-Koopman-generator} and Lemma~\ref{lm:dynamics-lifted}.
\end{proof}

The closed loop corresponding to the controller $\mu(x)$ obtained by Theorem~\ref{thm:stability-condition-LFR} is robust w.r.t.\ the uncertain but bounded remainder $\varepsilon(\hat{\Phi}(x),u)=\hat{r}(x,u)$ and the artificially introduced bounded uncertainty $\Delta=I_m\kron\hat{\Phi}(x)$. 
While the feedback design in Section~\ref{sec:controller-design-simple} leads to a robust controller which is linear in the lifted state $\hat{\Phi}(x)$, the controller~\eqref{eq:controller-nonlinear-explicit} depends nonlinearly on $\hat{\Phi}(x)$. 
This more general controller structure results from using that $\hat{\Phi}(x)$ and, thus, the uncertainty $\Delta$ in~\eqref{eq:dynamics-lifted} is not unknown but can be measured during operation to compute the feedback law.
We note that Theorem~\ref{thm:stability-condition-LFR} reduces to Theorem~\ref{thm:stability-condition-LFR-single-input} for $m=1$ and $L_w=0$, i.e., a feedback that is linear in the lifted state.
Due to its additional flexibility and more general controller structure, the design proposed in Theorem~\ref{thm:stability-condition-LFR}
allows for the consideration of larger error bounds $c_r$ on the unknown remainder or is feasible even for potentially fewer data samples.
Alternatively, the same amount of data can lead to a significantly increased region $\cX_\mathrm{RoA}$.
Despite the more flexible controller structure, the controller design relies on the solution of an SDP in terms of LMI feasibility conditions and the resulting computational complexity remains comparable to the one of the feedback design in Section~\ref{sec:controller-design-simple}. 
Hence, Theorem~\ref{thm:stability-condition-LFR} establishes an efficiently solvable and flexible controller design with probabilistic guarantees for the \emph{true} nonlinear system~\eqref{eq:dynamics-nonlinear} based on sampled data, where the necessary amount of data $d_0$ is specified by~\eqref{eq:mindata} for a given desired probabilistic tolerance $\delta\in (0,1)$ and a desired bound $c_r>0$ on the remainder.

The overall controller design based on EDMD~\eqref{eq:EDMD-optimization} and Theorem~\ref{thm:stability-condition-LFR} is summarized in Algorithm~\ref{alg:controller-design}.
\begin{algorithm}[tb]
    \caption{Flexible state-feedback controller design corresponding to Theorem~\ref{thm:stability-condition-LFR}.}\label{alg:controller-design}
    \textbf{Input:}
    \begin{itemize}
        \item Data $\{x_j^{\bar{u}}, \dot{x}_j^{\bar{u}}\}_{j=1}^{d^{\bar{u}}}$ for ${\bar{u}}\in\{0,e_1,...,e_m\}$, where $d^{\bar{u}}\geq d_0$ for $d_0$ according to~\eqref{eq:mindata}
        \item Lifting $\Phi(x) = \begin{bmatrix}1 & \hat{\Phi}(x)\end{bmatrix}^\top$ defined in~\eqref{eq:lifting-function}
        \item Probabilistic tolerance $\delta\in (0,1)$ and error bound $c_r > 0$
        \item Ellipsoidal region $\mathbf{\Delta}_\Phi$ defined by $Q_z\prec 0$, $S_z$, $R_z > 0$
    \end{itemize}
    \textbf{Output:} Exponentially stabilizing controller $u=\mu(x)$ for system~\eqref{eq:dynamics-nonlinear} with probability $1-\delta$ if $x\in\cX_\mathrm{RoA}$

\emph{Data-driven system representation:}
\begin{algorithmic}
    \State Arrange the data according to~\eqref{eq:data-matrices} in $X^{\bar{u}}, Y^{\bar{u}}$
    \State Solve the optimization problem~\eqref{eq:EDMD-optimization} to obtain the data-based system matrices $A$, $B_0$, $B_1$, ..., $B_m$ according to Lemma~\ref{lm:dynamics-lifted}
\end{algorithmic}

\emph{Controller design:}
\begin{algorithmic}
    \State Solve the LMI feasibility problem in Theorem~\ref{thm:stability-condition-LFR}
    \If{successful}
        \State Define the controller $\mu(x)$ as in~\eqref{eq:controller-nonlinear-explicit}
    \EndIf
    \end{algorithmic}
\end{algorithm}
\section{On the geometry of the region of attraction}\label{sec:geometry-of-RoA}
In this section, we discuss specific properties of the guaranteed RoA resulting from its quadratic parametrization in lifted coordinates.
Recall that our main results Theorem~\ref{thm:stability-condition-LFR-single-input} and Theorem~\ref{thm:stability-condition-LFR} guarantee exponential stability with a RoA given by a sublevel set of a quadratic Lyapunov function in the lifted space, i.e., $\cX_\mathrm{RoA} = \{x\in\bbR^n\,|\, \hat{\Phi}(x)^\top P^{-1} \hat{\Phi}(x) \leq 1\}$.
In the following, we discuss that a suitable choice of $P$ leading to a large RoA depends, as to be expected, on the distortion caused by the nonlinear lifting. 
In particular, we show that multiples of the identity matrix may not yield satisfactory results.
These insights have direct implications for determining a suitable choice of the matrices $Q_z$, $S_z$, $R_z$ in~\eqref{eq:condition-x-in-Z}, which parametrize the uncertainty description $\mathbf{\Delta}_\Phi$ characterizing the bilinear component of the dynamics.

As an instructive example, consider the lifting $\hat{\Phi}(x)=(x,x^2)$ with $x\in[-5,5]$.
Suppose we want to represent the set inclusion $x\in[-5,5]$ via the set $\mathbf{\Delta}_\Phi$, i.e., as a quadratic inequality in the lifted space.
This can be done by defining $\mathbf{\Delta}_\Phi$ such that $
    \mathbf{\Delta}_x 
    =\{
        (x,x^2), x\in[-5,5]
    \}
    \subseteq\mathbf{\Delta}_\Phi
$.
The simple choice $Q_{z,1}=-I$, $S_{z,1}=0$, $R_{z,1}=650$ leads to 
\begin{equation}
    \mathbf{\Delta}_{\Phi,1}
    =\left\{
        \Delta_\Phi\in\bbR^2
        \;\middle|\;
        \Delta_{\Phi,1}^2 + \Delta_{\Phi,2}^2 \leq 650
    \right\},
\end{equation}    
corresponding to the parametrization of $\mathbf{\Delta}_x$ as
\begin{equation}\label{eq:Delta-x-parametrization-simple}
    \mathbf{\Delta}_x
    =\left\{
        (x,x^2),\,x\in\bbR
        \;\middle|\;
        x^2 + x^4\leq 650
    \right\}.
\end{equation}
This is, however, a non-ideal representation as it does not account for the fact that $x^4$ significantly outweighs $x^2$ when $|x|>1$, and conversely for $|x|<1$.
Alternatively, we can, for example, rewrite 
\begin{equation}\label{eq:Delta-x-parametrization-balanced}
    \mathbf{\Delta}_x
    =\left\{
        (x,x^2),\,x\in\bbR
        \;\middle|\;
        \tfrac{1}{2}\left(5^{-2} x^2 + 5^{-4} x^4\right) \leq 1
    \right\}
\end{equation}
which leads with $Q_{z,2}=-\frac{1}{2}\diag(5^{-2},5^{-4})$, $S_{z,2}=0$, $R_{z,2}=1$ to a more balanced parametrization 
\begin{equation}
    \mathbf{\Delta}_{\Phi,2}
    =\left\{
        \Delta_\Phi\in\bbR^2
        \;\middle|\;
        \tfrac{1}{2}\left(5^{-2} \Delta_{\Phi,1}^2 + 5^{-4}\Delta_{\Phi,2}^2\right) \leq 1
    \right\}.
\end{equation}
Although both $\mathbf{\Delta}_{\Phi,1}$ and $\mathbf{\Delta}_{\Phi,2}$ contain $\mathbf{\Delta}_x$, the volume of $\mathbf{\Delta}_{\Phi,2}$ is significantly smaller than the volume of $\mathbf{\Delta}_{\Phi,1}$, cf.\ Fig.~\ref{fig:geometry-illustration}.
On the other hand, the nonlinear dependence of $\hat{\Phi}(x)=(x,x^2)\in\mathbf{\Delta}_x$, i.e., the second entry is always the first entry squared, causes the same set $\mathbf{\Delta}_x$ for both parametrizations~\eqref{eq:Delta-x-parametrization-simple} and~\eqref{eq:Delta-x-parametrization-balanced}.
This nonlinear dependence is, however, neglected during controller design, where we over-approximate the bilinearity using the ellipsoidal descriptions $\mathbf{\Delta}_{\Phi,1}\supseteq\mathbf{\Delta}_{x}$ and $\mathbf{\Delta}_{\Phi,2}\supseteq\mathbf{\Delta}_{x}$ as shown in Fig.~\ref{fig:geometry-illustration}.
\begin{figure}[t]
    \centering
    \begin{tikzpicture}[%
        /pgfplots/every axis y label/.style={at={(0,0.5)},xshift=-30pt,rotate=0},%
    ]%
        \begin{axis}[
            axis equal,
            xmin=-sqrt(650),xmax=sqrt(650),
            ymin=-sqrt(1250),ymax=sqrt(1250),
            legend pos= south east,
            xlabel=$x$,
            xtick distance=10,
            minor x tick num=9,
            ylabel=$x^2$,
            ytick distance=10,
            minor y tick num=9,
            grid=both,
            width = 0.95\columnwidth,
            minor grid style={gray!20},
            unbounded coords = jump
        ]
            \draw[thick, color=black, fill=black!70] (0,0) circle [x radius=sqrt(650), y radius=sqrt(650), rotate=0];
            \draw[thick, color=black, fill=black!10] (0,0) circle [x radius=sqrt(50), y radius=sqrt(1250), rotate=0];
            \draw[thick, color=black, dashed] (0,0) circle [x radius=sqrt(650), y radius=sqrt(650), rotate=0];
            \addplot[black, ultra thick,domain=-5:5] (x,x*x);
        \end{axis}
    \end{tikzpicture}
    \vspace*{-0.25\baselineskip}
    \setbox1=\hbox{\begin{tikzpicture}
        \draw[black,fill=black!70,thick](0,0) circle (0.08);
    \end{tikzpicture}}
    \setbox2=\hbox{\begin{tikzpicture}
        \draw[black,fill=black!10,thick](0,0) circle (0.08);
    \end{tikzpicture}}
    \setbox3=\hbox{\begin{tikzpicture}[baseline]
        \draw[black,ultra thick] (0,.6ex)--++(1.25em,0);
    \end{tikzpicture}}
    \caption{Illustration of the regions $\mathbf{\Delta}_x$ (\usebox3), $\mathbf{\Delta}_{\Phi,1}$ (\usebox1), and $\mathbf{\Delta}_{\Phi,2}$ (\usebox2).}
    \label{fig:geometry-illustration}
\vspace*{-0.5\baselineskip}
\end{figure}

For the proposed controller design, the above observation is especially crucial because $Q_{z}$, $S_{z}$, and $R_{z}$ are fixed a-priori and they define an over-approximation of the RoA.
In particular, the a-priori choice of $Q_{z}$ determines the shape of the set $\mathbf{\Delta}_\Phi$.
Suitable choices of $Q_z$ may allow for a larger RoA and should, thus, take the geometry of the RoA into account.

In the following, we propose a heuristic for determining a matrix $Q_z$ based on the chosen lifting and the resulting EDMD-based estimated system matrices, i.e., $A$, $B_0$, $B_1$ or $A$, $B_0$, $\tB$ corresponding to the steps for the data-driven system representation in Algorithm~\ref{alg:controller-design-single-input} or Algorithm~\ref{alg:controller-design}, respectively.
\begin{procedure}\label{proc:heuristic-Delta-Phi}
    The following steps lead to a non-trivial choice of $\mathbf{\Delta}_\Phi$ which takes the structure of the lifting into account:
    \begin{enumerate}
        \item Define $Q_z=-I$, $S_z=0$, and $R_z>0$.
        \item Find 
            \begin{enumerate}
                \item for Algorithm~\ref{alg:controller-design-single-input}: $\hP=\hP^\top\succ 0$, $\hL$, $\hat{\lambda}>0$, $\hat{\tau}>0$ such that~\eqref{eq:stability-condition-LFR-single-input}, or
                \item for Algorithm~\ref{alg:controller-design}: $\hP=\hP^\top\succ 0$, $\hL$, $\hL_w$, $\hat{\Lambda}=\hat{\Lambda}^\top\succ 0$, $\hat{\tau}>0$ such that~\eqref{eq:stability-condition-LFR}.
            \end{enumerate}
        \item Define $\mathbf{\Delta}_\Phi$ via $Q_z = -\frac{\hP^{-1}}{\|\hP^{-1}\|_2}$, $S_z=0$, and $R_z>0$.
    \end{enumerate}
\end{procedure}
After defining $\mathbf{\Delta}_\Phi$, we can proceed as stated in the controller design of Algorithm~\ref{alg:controller-design-single-input} or Algorithm~\ref{alg:controller-design}.
The steps in Procedure~\ref{proc:heuristic-Delta-Phi} solve the proposed controller design once without enforcing the RoA to be robust positively invariant. 
This results in the SDP-optimization having more degrees of freedom to find a suitable shape of the ellipsoidal RoA defined by $\hP$. 
As a consequence, the optimization can find a shape that naturally takes the dynamics \emph{and} the structure of the lifting function into account.
Further, it allows to define the uncertainty region $\mathbf{\Delta}_\Phi$ such that the shape of the obtained RoA $\cX_\mathrm{RoA}$ of the controller design is more similar to the shape of $\mathbf{\Delta}_\Phi$ itself and, thus, reduces conservatism in the robust controller design.
A rigorous investigation of the resulting RoA for a specific choice of $\mathbf{\Delta}_\Phi$ is a promising future path.
\section{Numerical examples}\label{sec:numerical-examples}
Next, we illustrate the proposed controller design in numerical examples. We conduct the simulations in Matlab using the toolbox YALMIP~\cite{lofberg:2004} with the SDP solver MOSEK~\cite{mosek:2022}.

\subsection{Nonlinear system with invariant Koopman lifting}\label{exmp:cooked-up-eigenfunctions}
    Consider the nonlinear system~(cf.~\cite{brunton:brunton:proctor:kutz:2016})
    \begin{equation}\label{eq:exmp-cooked-up}
        \dot{x}_1(t) = \rho x_1(t), 
        \quad
        \dot{x}_2(t) = \lambda(x_2(t) - x_1(t)^2) + u(t).
    \end{equation}
    with $\rho,\lambda \in \bbR$.
    To obtain a Koopman-based surrogate model, we define the lifting function~$\hat{\Phi}$ as 
    \begin{equation}
        \hat{\Phi}(x) = \begin{bmatrix}
            x_1 & x_2 & x_2 - \frac{\lambda}{\lambda-2\rho}x_1^2
        \end{bmatrix}^\top.
    \end{equation}
    A particular feature of~\eqref{eq:exmp-cooked-up} is that we can derive an exact finite-dimensional lifted bilinear representation given by
    \begin{align*}
        \ddt{}\hat{\Phi}(x(t)) 
        &= \left(\grad \hat{\Phi}(x(t))\right)^\top \left(\ddt{}x(t)\right)
        \\
        &= \begin{bmatrix}
            \rho & 0 & 0 \\
            0 & 2\rho & \lambda - 2 \rho \\
            0 & 0 & \lambda
        \end{bmatrix}
        \hat{\Phi}(x(t))
        + \begin{bmatrix}
            0 \\ 1 \\ 1
        \end{bmatrix}
        u(t).
    \end{align*}
    In the following, we choose $\rho=-2$ and $\lambda=1$.
    Since we assume that the system dynamics are unknown but we have only access to data samples, we use EDMD for a data-driven approximation. 
    For a data length of $d=\SI{5000}{}$, where the data samples are uniformly sampled from the sets $\bbX=[-1,1]^2$ and $\bbU=[-1,1]$, the EDMD optimization~\eqref{eq:EDMD-optimization} yields the data-based lifted system dynamics~\eqref{eq:dynamics-lifted-single-input} with 
    \begin{equation*}
        A = \left[\begin{smallmatrix}
           -2.0 & 0    & 0   \\
            0   &-4.0  & 5.0 \\
            0   & 0    & 1.0 
        \end{smallmatrix}\right]
        \!,\,
        B_0 = \left[\begin{smallmatrix}
            0   \\
            1.0 \\
            1.0
        \end{smallmatrix}\right]
        \!,\,
        B_1 = \left[\begin{smallmatrix}
             0   &  0  &   0 \\
             0   &  0  &   0 \\
             0   &  0  &   0
        \end{smallmatrix}\right]
        .
    \end{equation*}
    To account for a possible sampling error, we choose $c_r=\SI{0.1}{}$ and the probabilistic tolerance $\delta=\SI{0.05}{}$, cf. Proposition~\ref{prop:sampling-error-bound-Koopman-generator}. 
    Here, the choice of $d=\SI{5000}{}$ results in the approximators $A,B_0,B_1$ being accurate up to 13 digits.

    To apply the proposed controller design, we define the ellipsoidal region $\mathbf{\Delta}_\Phi$ via $Q_z = -I$, $S_z=0$, and $R_z=\SI{500}{}$. 
    Following Algorithm~\ref{alg:controller-design-single-input} for the simplified controller design, we obtain the state-feedback control law 
    \begin{equation*}
        \mu(x) = \begin{bmatrix}
             0.00 &  -3.77 &  -3.46
        \end{bmatrix}
        \hat{\Phi}(x).
    \end{equation*}
    Here, Algorithm~\ref{alg:controller-design} for the more flexible state-feedback controller yields $L_w=0$ and, hence, the same controller $\mu$.
    This can be explained by the linear structure of the lifted dynamics, allowing to achieve satisfactory control performance already with a linear state feedback.
    The above controller stabilizes the closed loop of~\eqref{eq:exmp-cooked-up} in the RoA depicted in Fig.~\ref{fig:exmp-cooked-up-eigenfunctions}.
    We note that the resulting RoA for this example is such that $\hat{\Phi}(x)\in\mathbf{\Delta}$ if and only if $x\in\cX_\mathrm{RoA}$, i.e., the region $\mathbf{\Delta}_\Phi$ characterizing the artificially introduced uncertainty only contains $x$ which are contained in the RoA as well. 
    Hence, the designed controller, which is robust w.r.t. to all $\Delta_\Phi\in\mathbf{\Delta}_\Phi$, is only robust w.r.t. $x\in\mathrm{RoA}$ and adds no conservatism by accounting also for some $x\not\in\cX_\mathrm{RoA}$.
    \begin{figure}[tb]
        \centering
        \begin{tikzpicture}[%
            /pgfplots/every axis x label/.style={at={(0.5,0)},yshift=-20pt},%
            /pgfplots/every axis y label/.style={at={(0,0.5)},xshift=-25pt,rotate=90},%
          ]%
            \begin{axis}[
                axis equal,
                legend pos= south east,
                xlabel=$x_1$,
                xmin=-20,
                xmax=20,
                xtick distance=5,
                minor x tick num=4,
                ylabel=$x_2$,
                ymin=-22,
                ymax=22,
                ytick distance=5,
                minor y tick num=4,
                grid=both,
                width = 0.95\columnwidth,
                minor grid style={gray!20},
                unbounded coords = jump
            ]
                \addplot[black,fill=black!40,thick,smooth] table [x index=4,y index=5] {data/exmp-cooked-up-eigenfunctions.dat};
            \end{axis}
        \end{tikzpicture}
        \vspace*{-0.25\baselineskip}
        \caption{RoA $\cX_\mathrm{RoA}$ corresponding to Section~\ref{exmp:cooked-up-eigenfunctions} and the controller $\mu$.}
        \label{fig:exmp-cooked-up-eigenfunctions}
    \end{figure}
    \begin{remark}
        For the considered system~\eqref{eq:exmp-cooked-up} and lifting $\hat{\Phi}$, Lemma~\ref{lm:dynamics-lifted} guarantees the upper bound~\eqref{eq:finite-gain-bound-remainder} on the remainder if $d\geq d_0$, where $d_0$ characterizes the sufficient number of data points defined in~\eqref{eq:mindata}.
        This bound can be evaluated if the system dynamics are known. 
        In particular, the bound yields $d_0 = \SI{6.9e17}{}$ for this example.
        We emphasize that this bound is clearly conservative and, as demonstrated above, already less data leads to an accurate data-driven approximation of the true system. 
        The conservatism is expected to be due to the validity of the bound for general system classes.
        For a discussion of faster rates, e.g., using Hoeffdings inequality and sharpness of the data requirements, we refer to~\cite{philipp:schaller:boshoff:peitz:nuske:worthmann:2024}.
    \end{remark}

\subsection{Nonlinear system without invariant Koopman lifting}\label{exmp:cooked-up-overapprox}
    The example in Section~\ref{exmp:cooked-up-eigenfunctions} is based on a lifting function that results in a perfect (bi-)linear representation of the dynamics. 
    Since this is rarely the case in practice, we include an additional lifting function in the following. 
    As before, we consider the nonlinear system~\eqref{eq:exmp-cooked-up} with $\rho=-2$ and $\lambda=1$, but additionally include the lifting function $\phi_4(x)=x_1x_2$ resulting in the lifting function 
    \begin{equation*}
        \hat{\Phi}(x) = \begin{bmatrix}
            x_1 & x_2 & x_2 - \frac{\lambda}{\lambda-2\rho}x_1^2 & x_1x_2
        \end{bmatrix}^\top.
    \end{equation*}
    Note that the corresponding dictionary is not invariant under the system dynamics~\eqref{eq:exmp-cooked-up} and thereby violates Assumption~\ref{ass:invariance-of-dictionary}.
    In the following, however, we show that the proposed approach still produces reliable results.
    For the data generation, we sample again $d=\SI{5000}{}$ data points uniformly from the interval $\bbX=[-1,1]^2$ and $\bbU=[-1,1]$. To illustrate the practical usability of our approach, we consider the case of imperfect state derivative measurements, i.e., we consider the data $\{x_j,\dot{\tilde{x}}_j\}_{j=1}^{d}$ with $\dot{\tilde{x}}_j = \dot{x}_j + \xi_j$, where $\xi_j$ is uniformly sampled from the interval $[-\bar{\xi},\bar{\xi}]^2$ with $\bar{\xi}=0.05$. Following Section~\ref{sec:DD-system-representation} yields a lifted bilinear representation~\eqref{eq:dynamics-lifted-single-input}, where we assume that the remainder is bounded by~\eqref{eq:finite-gain-bound-remainder} with $c_r=\SI{0.01}{}$ and probabilistic tolerance $\delta=\SI{0.05}{}$.
    Further, we implement Algorithm~\ref{alg:controller-design} for the ellipsoidal region $\mathbf{\Delta}_\Phi$ based on $Q_z = -I$, $S_z=0$, and $R_z=\SI{1000}{}$ leading to the controller 
    \begin{equation*}
        \mu_1(x) = \frac{
            \begin{bmatrix}
                 -0.0001  & -0.0132  & -2.8888   & 0.0001
            \end{bmatrix}
            \hat{\Phi}(x)
        }{
            1 - \begin{bmatrix}
                 0.0001 & -0.0000 & -0.0000 & -0.0000
            \end{bmatrix}
            \hat{\Phi}(x)
        }
    \end{equation*}
    stabilizing the closed-loop of~\eqref{eq:exmp-cooked-up} in the RoA depicted in Fig.~\ref{fig:exmp-cooked-up-overapprox}.
    \begin{figure}[tb]
        \centering
        \begin{tikzpicture}[%
            /pgfplots/every axis x label/.style={at={(0.5,0)},yshift=-20pt},%
            /pgfplots/every axis y label/.style={at={(0,0.5)},xshift=-25pt,rotate=90},%
          ]%
            \begin{axis}[
                axis equal,
                legend pos= south east,
                xlabel=$x_1$,
                xmin=-20,
                xmax=20,
                xtick distance=5,
                minor x tick num=4,
                ylabel=$x_2$,
                ymin=-22,
                ymax=22,
                ytick distance=5,
                minor y tick num=4,
                grid=both,
                width = 0.95\columnwidth,
                minor grid style={gray!20},
                unbounded coords = jump
            ]
                \addplot[black,fill=black!70,thick,smooth] table [x index=0,y index=1] {data/exmp-cooked-up-overapprox.dat};
                \addplot[black,fill=black!40,thick,smooth] table [x index=4,y index=5] {data/exmp-cooked-up-overapprox.dat};
            \end{axis}
        \end{tikzpicture}
        \vspace*{-0.25\baselineskip}
        \setbox1=\hbox{\begin{tikzpicture}
            \draw[black,fill=black!70,thick](0,0) circle (0.08);
        \end{tikzpicture}}
        \setbox2=\hbox{\begin{tikzpicture}
            \draw[black,fill=black!40,thick](0,0) circle (0.08);
        \end{tikzpicture}}
        \caption{Region containing all $x$ with $\hat{\Phi}(x)\in\mathbf{\Delta}_\Phi$ (\usebox1) corresponding to Section~\ref{exmp:cooked-up-overapprox} and the RoA $\cX_\mathrm{RoA}$ (\usebox2) for the controller $\mu_1$.}
        \label{fig:exmp-cooked-up-overapprox}
    \end{figure}
    Recall that the region $\mathbf{\Delta}_\Phi$ is chosen and fixed prior to the controller design. 
    The controller design then optimizes for a controller guaranteeing closed-loop stability for a large RoA.
    Ideally, this RoA contains all $x$ with $\hat{\Phi}(x)\in\mathbf{\Delta}_\Phi$ and, hence, a good a-priori choice of $\Delta_\Phi$ is critical (cf. Section~\ref{sec:geometry-of-RoA}).
    Whereas this was indeed achieved by the controller in Section~\ref{exmp:cooked-up-eigenfunctions}, the resulting RoA in Fig.~\ref{fig:exmp-cooked-up-overapprox} contains only a small portion of all $x\in \bbR^2$ satisfying $\hat{\Phi}\in\mathbf{\Delta}_\Phi$.
    This is a clear indication that the chosen $\mathbf{\Delta}_\Phi$ is not optimal w.r.t. the used lifting functions and the underlying nonlinear dynamics~\eqref{eq:exmp-cooked-up}.
    
    To investigate this further, we consider the same setting as before but with a different choice of $\mathbf{\Delta}_\Phi$.
    In particular, we define $\mathbf{\Delta}_\Phi$ again via $S_z=0$ and $R_z=\SI{1000}{}$, but with the more sophisticated choice of $Q_z = -\diag(2.5,2.5,1.25,0.005)$ with an individual weight for each lifting function. 
    Then, Algorithm~\ref{alg:controller-design} yields the controller
    \begin{equation*}
        \mu_2(x) = \frac{
            \begin{bmatrix}
                 -0.0049 & -2.0567 & -2.4020 & -0.0092
            \end{bmatrix}
            \hat{\Phi}(x)
        }{
            1 - \begin{bmatrix}
                 0.0053 & -0.0000 & -0.0000 & -0.0001
            \end{bmatrix}
            \hat{\Phi}(x)
        }
    \end{equation*}
    stabilizing the closed-loop of~\eqref{eq:exmp-cooked-up} in the RoA depicted in Fig.~\ref{fig:exmp-cooked-up-overapprox-tuned}.
    \begin{figure}[tb]
        \centering
        \begin{tikzpicture}[%
            /pgfplots/every axis x label/.style={at={(0.5,0)},yshift=-20pt},%
            /pgfplots/every axis y label/.style={at={(0,0.5)},xshift=-25pt,rotate=90},%
          ]%
            \begin{axis}[
                axis equal,
                legend pos= south east,
                xlabel=$x_1$,
                xmin=-20,
                xmax=20,
                xtick distance=5,
                minor x tick num=4,
                ylabel=$x_2$,
                ymin=-22,
                ymax=22,
                ytick distance=5,
                minor y tick num=4,
                grid=both,
                width = 0.95\columnwidth,
                minor grid style={gray!20},
                unbounded coords = jump
            ]
                \addplot[black,fill=black!70,thick,smooth] table [x index=0,y index=1] {data/exmp-cooked-up-overapprox-tuned.dat};
                \addplot[black,fill=black!40,thick,smooth] table [x index=4,y index=5] {data/exmp-cooked-up-overapprox-tuned.dat};
            \end{axis}
        \end{tikzpicture}
        \vspace*{-0.25\baselineskip}
        \setbox1=\hbox{\begin{tikzpicture}
            \draw[black,fill=black!70,thick](0,0) circle (0.08);
        \end{tikzpicture}}
        \setbox2=\hbox{\begin{tikzpicture}
            \draw[black,fill=black!40,thick](0,0) circle (0.08);
        \end{tikzpicture}}
        \caption{Region containing all $x$ with $\hat{\Phi}(x)\in\mathbf{\Delta}_\Phi$ (\usebox1) corresponding to Section~\ref{exmp:cooked-up-overapprox} and the RoA $\cX_\mathrm{RoA}$ (\usebox2) for the controller $\mu_2$.}
        \label{fig:exmp-cooked-up-overapprox-tuned}
    \end{figure}
    We stress that this RoA a) is significantly larger and b) provides a more accurate cover of the region of all $x$ with $\hat{\Phi}(x)\in\mathbf{\Delta}_\Phi$ than the one for $\mathbf{\Delta}_\Phi$ based on $Q_z=-I$.
    This illustrates the important role of the shape of $\mathbf{\Delta}_\Phi$ in the proposed controller design.

    We note that the obtained RoA is only guaranteed for perfect measurements, but the above simulation studies for noisy data showed comparable results as in the noise-free case. The theoretical analysis for our proposed framework with noisy data is left for future work.
\subsection{Inverted pendulum}\label{exmp:inverse-pendulum-sin}
    To illustrate our controller design with an actually physical example, we investigate an inverted pendulum which is a well-used benchmark for nonlinear controller design methods~(cf.~\cite{strasser:berberich:allgower:2021,tiwari:nehma:lusch:2023,martin:schon:allgower:2023a,verhoek:abbas:toth:2023,martin:schon:allgower:2023b} and references therein).
    To this end, we consider the dynamics 
    \begin{subequations}\label{eq:exmp-inverse-pendulum-dynamics}
        \begin{align}
            \dot{x}_1(t) &= x_2(t), \\
            \dot{x}_2(t) &= 
            \tfrac{g}{l}\sin(x_1(t))
            - \tfrac{b}{ml^2} x_2(t)
            + \tfrac{1}{ml^2} u(t)
        \end{align}
    \end{subequations}
    with mass $m$, length $l$, rotational friction coefficient $b$, and gravitational constant $g=\SI{9.81}{m\per s}$. 
    For the simulation, we choose $m=1$, $l=1$, and $b=0.01$.
    We define the lifting function
    $
        \hat{\Phi}(x) = \begin{bmatrix}
            x_1 & x_2 & \sin(x_1)
        \end{bmatrix}^\top
    $
    and use $d=\SI{15000}{}$ data points uniformly sampled from $\bbX=[-2,10]^2$ and $\bbU=[-10,10]$.
    This leads to the data-driven bilinear Koopman surrogate model~\eqref{eq:dynamics-lifted-single-input} with 
    \begin{gather*}
        A \!=\! \left[\begin{smallmatrix}
             0    &  1.00 &  0    \\
             0    & -0.50 &  9.81 \\
            -0.17 &  0.13 &  0.03
        \end{smallmatrix}\right] 
        \!,\,
        B_0 \!=\! \left[\begin{smallmatrix}
            0    \\
            1.00 \\
            1.12
        \end{smallmatrix}\right] 
        \!,\,
        B_1 \!=\! \left[\begin{smallmatrix}
             0    &  0    &  0    \\
             0    &  0    &  0    \\
            -0.10 & -0.10 &  0.13
        \end{smallmatrix}\right] 
        \!,
    \end{gather*}
    where the remainder $\hat{r}(x,u)$ is bounded by~\eqref{eq:finite-gain-bound-remainder} with $c_r=\SI{0.02}{}$ and probabilistic tolerance $\delta=\SI{0.05}{}$.
    In the following, we compare two different choices for the region $\mathbf{\Delta}_\Phi$ and their implications for the application of Algorithms~\ref{alg:controller-design-single-input} and~\ref{alg:controller-design}.

    1) First, we use the straightforward definition of $\mathbf{\Delta}_\Phi$ via $Q_z = -I$, $S_z=0$, and $R_z=\SI{30}{}$. 
    Then, Algorithm~\ref{alg:controller-design-single-input} and Algorithm~\ref{alg:controller-design} yield the simplified state feedback
    \begin{equation*}
        \mu_1(x) = \begin{bmatrix}
            -1.4095 & -8.1208 & -8.8105
        \end{bmatrix}
        \hat{\Phi}(x)
    \end{equation*}
    and the more flexible state feedback
    \begin{equation*}
        \mu_2(x) = \frac{
            \begin{bmatrix}
                 -0.5967 & -6.5454 & -4.6422
            \end{bmatrix}
            \hat{\Phi}(x)
        }{
            1 - \begin{bmatrix}
                 0.0453 & 0.0435 & -0.0579
            \end{bmatrix}
            \hat{\Phi}(x)
        }
    \end{equation*}
    with the respective RoA depicted in Fig.~\ref{fig:exmp-inverse-pendulum-sin}.
    \begin{figure}[t]
        \centering
        \begin{tikzpicture}[%
            /pgfplots/every axis x label/.style={at={(0.5,0)},yshift=-20pt},%
            /pgfplots/every axis y label/.style={at={(0,0.5)},xshift=-25pt,rotate=90},%
          ]%
            \begin{axis}[
                axis equal,
                legend pos= south east,
                xlabel=$x_1$,
                xmin=-10,
                xmax=10,
                xtick distance=2,
                minor x tick num=3,
                ylabel=$x_2$,
                ymin=-10,
                ymax=10,
                ytick distance=2,
                minor y tick num=3,
                grid=both,
                width = 0.95\columnwidth,
                minor grid style={gray!20},
                unbounded coords = jump
            ]
                \addplot[black,fill=black!70,thick,smooth] table [x index=0,y index=1] {data/exmp-inverse-pendulum-sin.dat};
                \addplot[black,fill=black!40,thick,smooth] table [x index=4,y index=5] {data/exmp-inverse-pendulum-sin.dat};
                \addplot[black,fill=black!10,thick,smooth] table [x index=2,y index=3] {data/exmp-inverse-pendulum-sin.dat};
            \end{axis}
        \end{tikzpicture}
        \vspace*{-0.25\baselineskip}
        \setbox1=\hbox{\begin{tikzpicture}
            \draw[black,fill=black!70,thick](0,0) circle (0.08);
        \end{tikzpicture}}
        \setbox2=\hbox{\begin{tikzpicture}
            \draw[black,fill=black!40,thick](0,0) circle (0.08);
        \end{tikzpicture}}
        \setbox3=\hbox{\begin{tikzpicture}
            \draw[black,fill=black!10,thick](0,0) circle (0.08);
        \end{tikzpicture}}
        \caption{Region containing all $x$ with $\hat{\Phi}(x)\in\mathbf{\Delta}_\Phi$ (\usebox1) corresponding to Section~\ref{exmp:inverse-pendulum-sin} and the RoA $\cX_\mathrm{RoA}$ for the controllers $\mu_1$ (\usebox3) and $\mu_2$ (\usebox2).}
        \label{fig:exmp-inverse-pendulum-sin}
    \end{figure}
    While the resulting RoAs acceptably cover the region corresponding to $\mathbf{\Delta}_\Phi$, the RoA for both $\mu_1$ and $\mu_2$ are ellipsoidal.
    This indicates that the ellipsoidal region $\mathbf{\Delta}_\Phi$ for $\hat{\Phi}(x)$ constrains mostly the state $x$ itself. 
    Hence, the presented robust controller design has to account for a possibly unnecessarily large uncertainty region.

    2) As a remedy, we define $\mathbf{\Delta}_\Phi$ now based on the proposed heuristic in Section~\ref{sec:geometry-of-RoA}.
    In particular, we first follow the steps of Procedure~\ref{proc:heuristic-Delta-Phi} to obtain a non-trivial choice of $Q_z$.
    This matrix is then used to define the shape of $\mathbf{\Delta}_\Phi$. 
    As before, we choose $S_z=0$ and use $R_z=5$ to scale the overall size of the region.
    Applying Algorithms~\ref{alg:controller-design-single-input} and~\ref{alg:controller-design} yields the controllers
    \begin{align}
        \mu_3(x) &= \begin{bmatrix}
            -3.2179 & -2.8083 & -13.6730
        \end{bmatrix}
        \hat{\Phi}(x),
        \\
        \mu_4(x) &= \frac{
            \begin{bmatrix}
                 -1.6091 & -1.3541 & -8.3180
            \end{bmatrix}
            \hat{\Phi}(x)
        }{
            1 - \begin{bmatrix}
                 0.0493 & 0.0469 & -0.0610
            \end{bmatrix}
            \hat{\Phi}(x)
        },
    \end{align}
    respectively, stabilizing the closed-loop of~\eqref{eq:exmp-inverse-pendulum-dynamics} in the significantly larger RoA depicted in Fig.~\ref{fig:exmp-inverse-pendulum-sin-tuned}.

    To validate the effectiveness of our method, we compute closed-loop trajectories for both controllers $\mu_3$ and $\mu_4$ and compare it to a linear-quadratic regulator $\mu_\mathrm{LQR}$ for a \emph{linear} Koopman surrogate model (cf.~\cite{brunton:brunton:proctor:kutz:2016}) using the same data and the same lifting function $\hat{\Phi}(x)$. 
    Whereas $\mu_\mathrm{LQR}$ is not able to stabilize the origin of the nonlinear system, both controllers $\mu_3$ and $\mu_4$ render the origin closed-loop stable.
    \begin{figure}[t]
        \centering
        \begin{tikzpicture}[%
            /pgfplots/every axis x label/.style={at={(0.5,0)},yshift=-20pt},%
            /pgfplots/every axis y label/.style={at={(0,0.5)},xshift=-25pt,rotate=90},%
          ]%
            \begin{axis}[
                axis equal,
                legend pos= south east,
                xlabel=$x_1$,
                xmin=-10,
                xmax=10,
                xtick distance=2,
                minor x tick num=3,
                ylabel=$x_2$,
                ymin=-10,
                ymax=10,
                ytick distance=2,
                minor y tick num=3,
                grid=both,
                width = 0.95\columnwidth,
                minor grid style={gray!20},
                unbounded coords = jump
            ]
                \addplot[black,fill=black!70,thick,smooth] table [x index=0,y index=1] {data/exmp-inverse-pendulum-sin-tuned.dat};
                \addplot[black,fill=black!40,thick,smooth] table [x index=4,y index=5] {data/exmp-inverse-pendulum-sin-tuned.dat};
                \addplot[black,fill=black!10,thick,smooth] table [x index=2,y index=3] {data/exmp-inverse-pendulum-sin-tuned.dat};
                \addplot[black,thick,smooth,dashed] table [x index=6,y index=7] {data/exmp-inverse-pendulum-sin-tuned.dat};
                \addplot[black,thick,smooth,dotted] table [x index=8,y index=9] {data/exmp-inverse-pendulum-sin-tuned.dat};
                \addplot[black,thick,smooth] table [x index=10,y index=11] {data/exmp-inverse-pendulum-sin-tuned.dat};
                %
                \addplot[black,thick,smooth,dashed] table [x index=18,y index=19] {data/exmp-inverse-pendulum-sin-tuned.dat};
                \addplot[black,thick,smooth,dotted] table [x index=20,y index=21] {data/exmp-inverse-pendulum-sin-tuned.dat};
                \addplot[black,thick,smooth] table [x index=22,y index=23] {data/exmp-inverse-pendulum-sin-tuned.dat};
                \addplot[black,thick,smooth,dashed] table [x index=24,y index=25] {data/exmp-inverse-pendulum-sin-tuned.dat};
                \addplot[black,thick,smooth,dotted] table [x index=26,y index=27] {data/exmp-inverse-pendulum-sin-tuned.dat};
                \addplot[black,thick,smooth] table [x index=28,y index=29] {data/exmp-inverse-pendulum-sin-tuned.dat};
                \addplot[black,thick,smooth,dashed] table [x index=30,y index=31] {data/exmp-inverse-pendulum-sin-tuned.dat};
                \addplot[black,thick,smooth,dotted] table [x index=32,y index=33] {data/exmp-inverse-pendulum-sin-tuned.dat};
                \addplot[black,thick,smooth] table [x index=34,y index=35] {data/exmp-inverse-pendulum-sin-tuned.dat};
                %
            \end{axis}
        \end{tikzpicture}
        \vspace*{-1.25\baselineskip}
        \setbox1=\hbox{\begin{tikzpicture}
            \draw[black,fill=black!70,thick](0,0) circle (0.08);
        \end{tikzpicture}}
        \setbox2=\hbox{\begin{tikzpicture}
            \draw[black,fill=black!40,thick](0,0) circle (0.08);
        \end{tikzpicture}}
        \setbox3=\hbox{\begin{tikzpicture}
            \draw[black,fill=black!10,thick](0,0) circle (0.08);
        \end{tikzpicture}}
        \setbox4=\hbox{\begin{tikzpicture}[baseline]
            \draw[black,thick,smooth,dashed] (0,.6ex)--++(0.9em,0);
        \end{tikzpicture}}
        \setbox5=\hbox{\begin{tikzpicture}[baseline]
            \draw[black,thick,smooth,dotted] (0,.6ex)--++(0.9em,0);
        \end{tikzpicture}}
        \setbox6=\hbox{\begin{tikzpicture}[baseline]
            \draw[black,thick,smooth] (0,.6ex)--++(0.9em,0);
        \end{tikzpicture}}
        \caption{Region containing all $x$ with $\hat{\Phi}(x)\in\mathbf{\Delta}_\Phi$ (\usebox1) corresponding to Section~\ref{exmp:inverse-pendulum-sin}, the RoA $\cX_\mathrm{RoA}$ for the controllers $\mu_3$ (\usebox3) and $\mu_4$ (\usebox2), as well as closed-loop trajectories for $\mu_3$ (\usebox4), $\mu_4$ (\usebox5), and $\mu_\mathrm{LQR}$ (\usebox6).}
        \label{fig:exmp-inverse-pendulum-sin-tuned}
        \vspace{-0.7\baselineskip}
    \end{figure}%
\section{Conclusion}\label{sec:conclusion}
In this paper, we presented a robust controller design to overcome challenges posed by finite-dimensional and data-driven approximations of Koopman surrogate models. 
In particular, we designed a flexible state-feedback controller based on the lifted system dynamics ensuring end-to-end stability guarantees for the true nonlinear system.
Moreover, our method established a crucial connection between closed-loop guarantees and the required data for a probabilistic bound on the approximation error.
Further, we used semidefinite programming and linear matrix inequalities to ensure an efficient and reliable design procedure.
The tutorial-style explanation of our results in Section~\ref{sec:controller-design-simple} enhanced accessibility and our extension to more flexible nonlinear state-feedback controllers in Section~\ref{sec:controller-design-general} reduced the conservatism of the approach. 
The broad applicability of the presented controller design framework was illustrated in numerical examples.

In future work, we aim to explicitly address the projection error due to an approximate and finite set of dictionary functions, considering a potential transition to a reproducing Kernel Hilbert space setting~\cite{philipp:schaller:worthmann:peitz:nuske:2024}. 
Additionally, we plan to explore alternative sampling strategies beyond the currently employed uniform sampling strategy, e.g., using trajectory data for specific inputs.
Possible future research includes investigating nonlinear systems with multiple equilibrium points~\cite{liu:ozay:sontag:2023}.

\bibliographystyle{IEEEtran}
\bibliography{literature}

\begin{IEEEbiography}[{\includegraphics[width=1in,height=1.25in,clip,keepaspectratio]{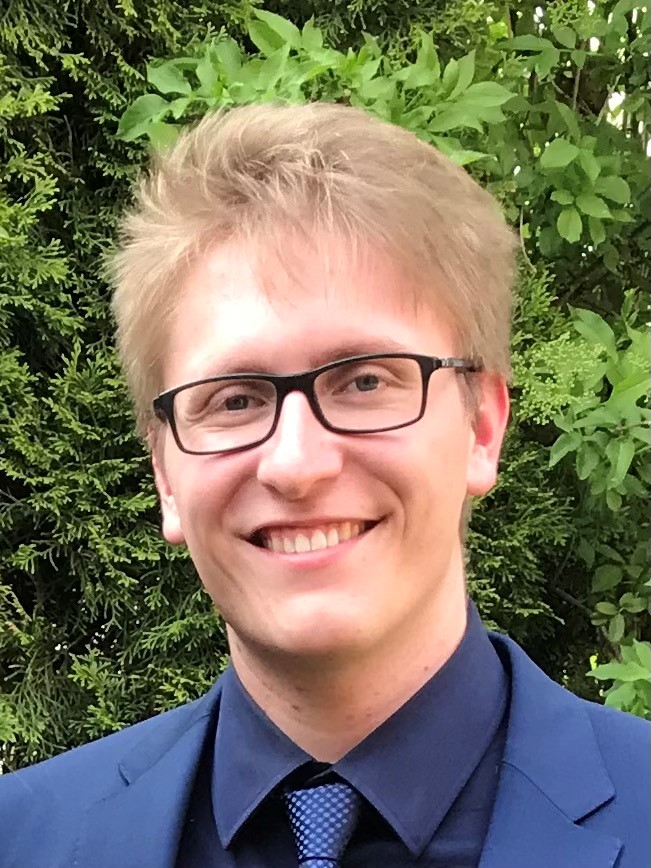}}]{Robin Strässer} received a master’s degree in Simulation Technology from the University of Stuttgart, Stuttgart, Germany, in 2020.
Since 2020, he has been a Research and Teaching Assistant with the Institute for Systems Theory and Automatic Control and a member of the Graduate School Simulation Technology, University of Stuttgart. 
His research interests include data-driven system analysis and control with focus on nonlinear systems.
Robin Strässer received the Best Poster Award at the International Conference on Data-Integrated Simulation Science (SimTech2023).
\end{IEEEbiography}

\begin{IEEEbiography}[{\includegraphics[width=1in,height=1.25in,clip,keepaspectratio]{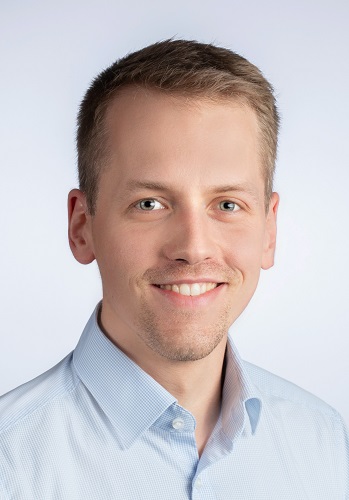}}]{Manuel Schaller} obtained the M.Sc.\ and Ph.D.\ in Applied Mathematics from the University of Bayreuth in 2017 and 2021 respectively. From 2020-2023 he held a PostDoc and Lecturer (tenure track) position in the Optimization-based Control group at Technische Universität Ilmenau, Germany. There, he is Assistant Professor for differential equations since July 2023. His research focuses on data-driven control with guarantees, port-Hamiltonian systems and stability in infinite dimensional optimal control.
For his research Dr.\ Schaller has been named junior fellow of the GAMM (Society for Applied Mathematics and Mechanics) and received the Best Poster Award at the workshop on systems theory and PDEs (WOSTAP 2022). 
\end{IEEEbiography}

\begin{IEEEbiography}[{\includegraphics[width=1in,height=1.25in,clip,keepaspectratio]{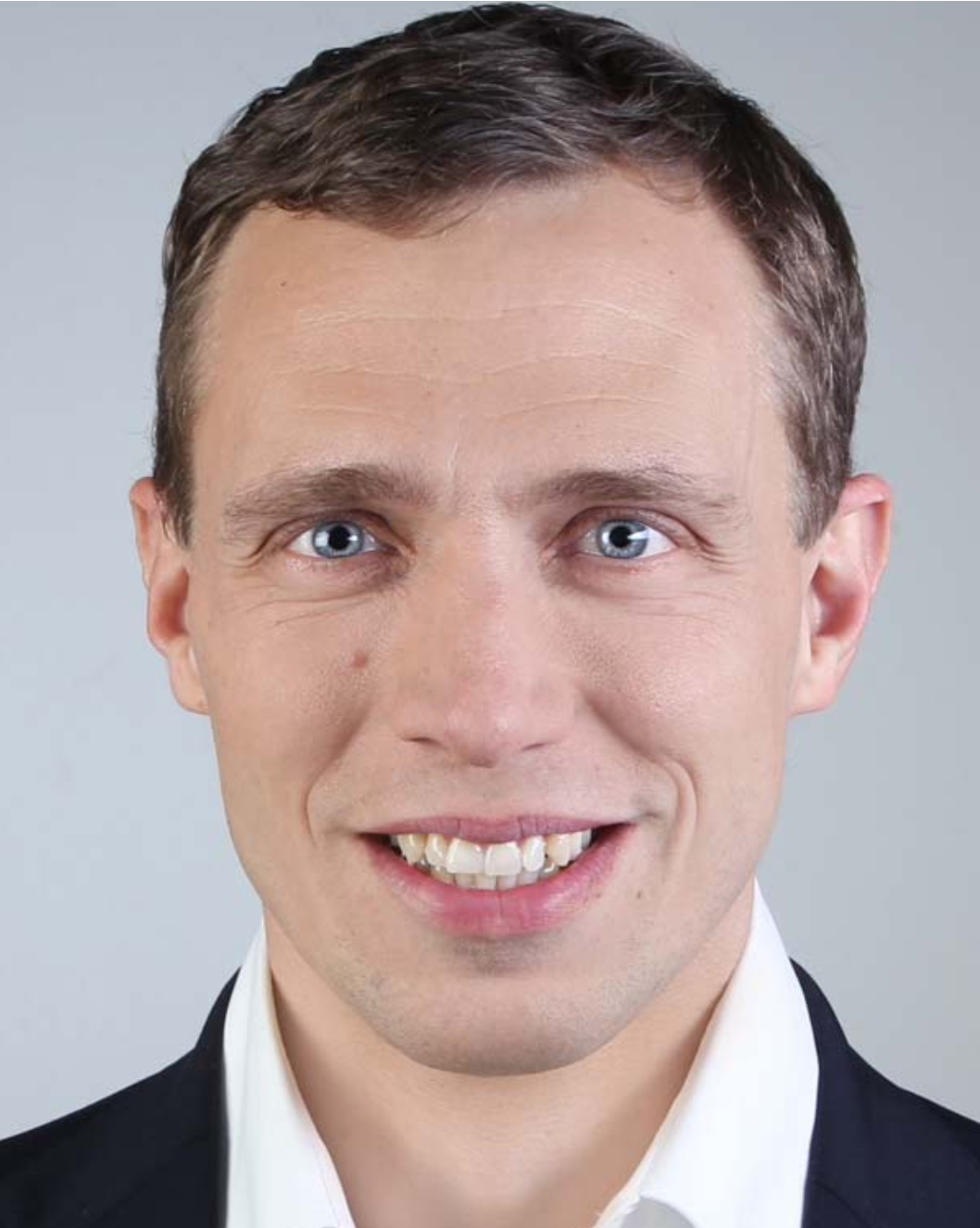}}]{Karl Worthmann} received his Ph.D.\ degree in mathematics from the University of Bayreuth, Germany, in 2012. 2014 he become assistant professor for ''Differential Equations'' at Technische Universität Ilmenau (TU Ilmenau), Germany. 2019 he was promoted to full professor after receiving the Heisenberg-professorship ''Optimization-based Control'' by the German Research Foundation (DFG). He was recipient of the Ph.D.\ Award from the City of Bayreuth, Germany, and stipend of the German National Academic Foundation. 2013 he has been appointed Junior Fellow of the Society of Applied Mathematics and Mechanics (GAMM), where he served as speaker in 2014 and 2015. 
His current research interests include systems and control theory with a particular focus on nonlinear model predictive control, stability analysis, and data-driven control.
\end{IEEEbiography}

\begin{IEEEbiography}[{\includegraphics[width=1in,height=1.25in,clip,keepaspectratio]{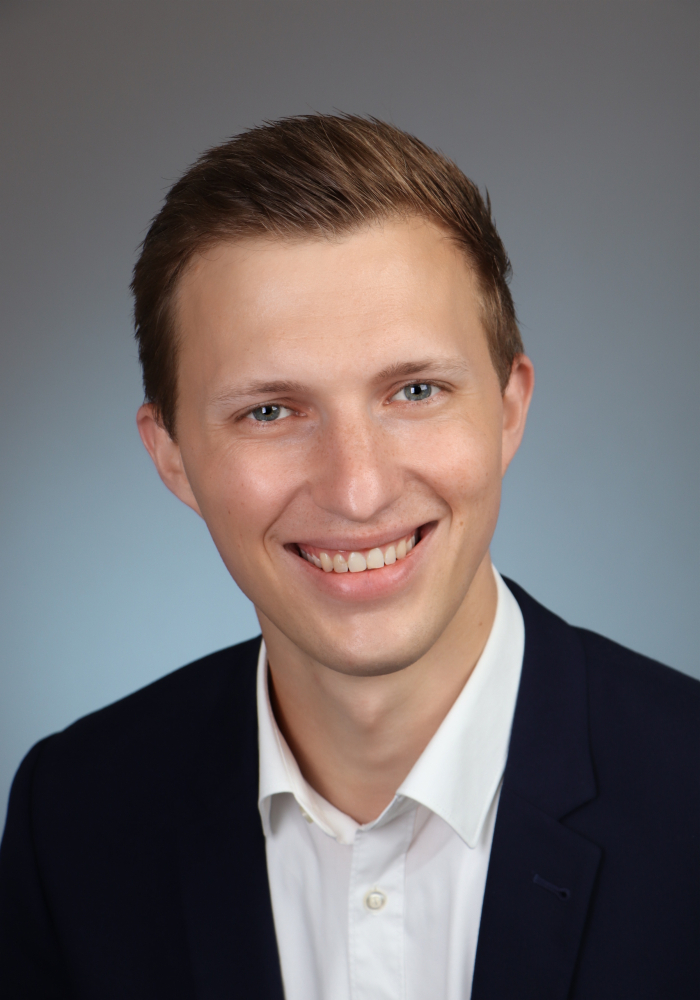}}]{Julian Berberich} received an M.Sc. degree in Engineering Cybernetics from the University of Stuttgart, Germany, in 2018.
In 2022, he obtained a Ph.D. in Mechanical Engineering, also from the University of Stuttgart, Germany. 
He is currently working as a Lecturer (Akademischer Rat) at the Institute for Systems Theory and Automatic Control at the University of Stuttgart, Germany. 
In 2022, he was a visiting researcher at the ETH Zürich, Switzerland. 
He has received the Outstanding Student Paper Award at the 59th IEEE Conference on Decision and Control in 2020 and the 2022 George S. Axelby Outstanding Paper Award. 
His research interests include data-driven analysis and control as well as quantum computing.
\end{IEEEbiography}

\begin{IEEEbiography}[{\includegraphics[width=1in,height=1.25in,clip,keepaspectratio]{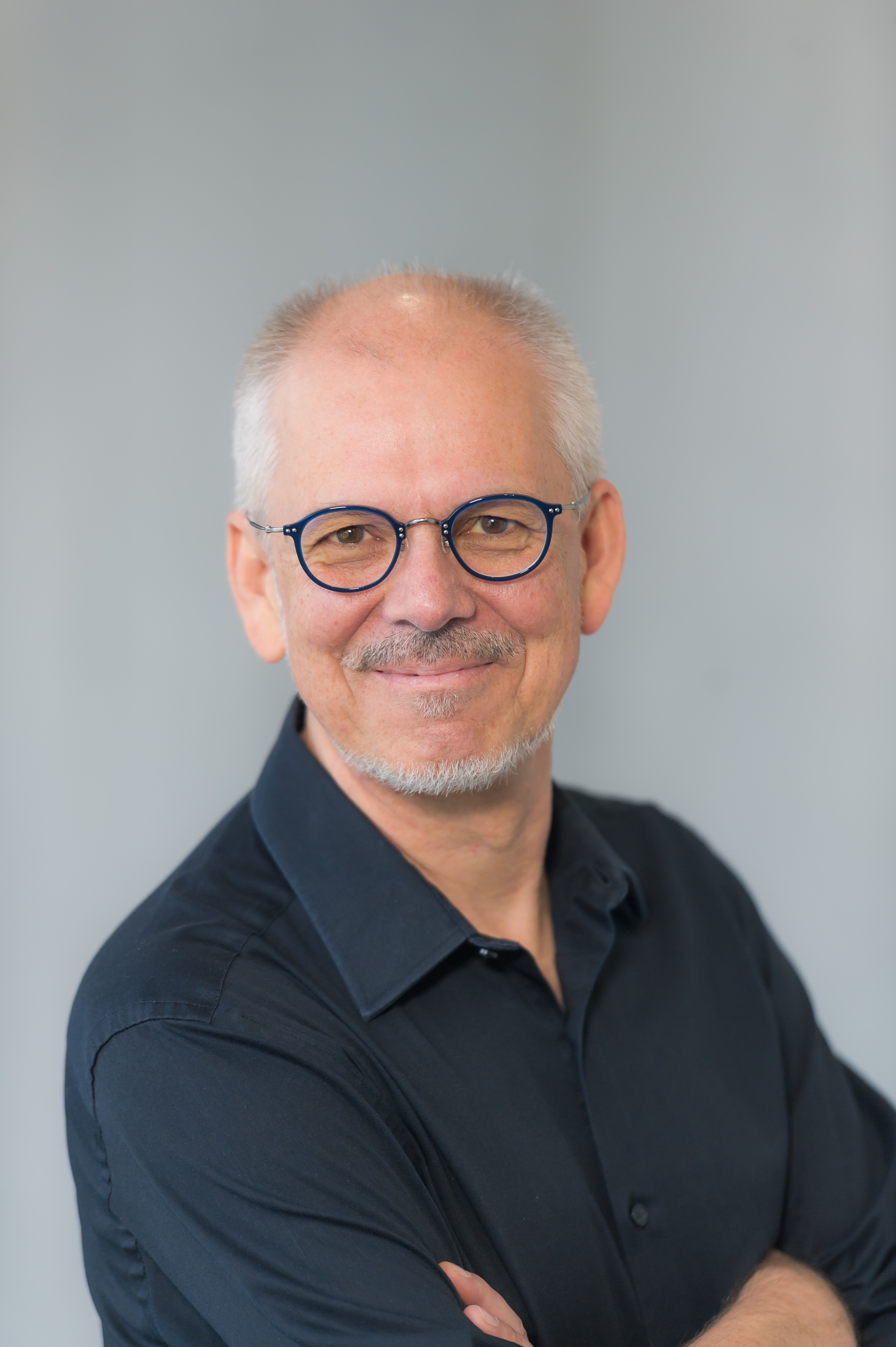}}]{Frank Allgöwer} studied engineering cybernetics and applied mathematics in Stuttgart and with the University of California, Los Angeles (UCLA), CA, USA, respectively, and received the Ph.D. degree from the University of Stuttgart, Stuttgart, Germany.
Since 1999, he has been the Director of the Institute for Systems Theory and Automatic Control and a professor with the University of Stuttgart. 
His research interests include predictive control, data-based control, networked control, cooperative control, and nonlinear control with application to a wide range of fields including systems biology.
Dr. Allgöwer was the President of the International Federation of Automatic Control (IFAC) in 2017–2020 and the Vice President of the German Research Foundation DFG in 2012–2020.
\end{IEEEbiography}

\end{document}